\documentclass[a4paper, 11pt]{article}
\usepackage{cite}
\usepackage[top=1.0 in,bottom=1.0 in,left=1.0 in,right=1.0 in]{geometry}
\usepackage{amsmath,amsthm}
\usepackage{amssymb}
\usepackage[colorlinks,linkcolor=blue,citecolor=blue]{hyperref}
\usepackage{float}
\usepackage{cite}
\usepackage{epic}
\usepackage{tikz}
\usepackage{graphicx}
\usepackage{subfigure}
\usepackage{accents}
\usepackage{slashed}
\newtheorem{definition}{Definition}[section]

\newtheorem{theorem}{Theorem}[section]
\newtheorem{Proposition}{Proposition}[section]

\numberwithin{equation}{section}

\newcommand{\p}{\prime}

\newcommand{\mf}{\mathbf}

\makeatletter

\newcommand{\Rmnum}[1]{\expandafter\@slowromancap\romannumeral #1@} \makeatother

\title{The direct linearization scheme with the Lam\'e function: The KP equation and reductions}

\author{Xing Li$^1$,~~ Ying-ying Sun$^2$,~~ Da-jun Zhang$^{3,4}$\footnote{
Corresponding author. Email: djzhang@staff.shu.edu.cn}\\
{\small $^1$School of Mathematics and Statistics, Jiangsu Normal University, Xuzhou 221116, China}\\
{\small $^2$Department of Mathematics, University of Shanghai for Science and Technology, Shanghai 200093, China}\\
{\small $^3$Department  of Mathematics, Shanghai University, Shanghai 200444, China}\\
{\small $^{4}$Newtouch Center for Mathematics of Shanghai University,  Shanghai 200444, China}}

\date{\today}

\begin{document}

\maketitle

\begin{abstract}
The paper starts from establishing an elliptic direct linearization (DL) scheme for
the Kadomtsev-Petviashvili equation.
The scheme consists of an integral equation (involving the Lam\'e function)
and a formula for elliptic soliton solutions,
which can be confirmed by checking Lax pair.
Based on analysis of real-valuedness of the Weierstrass functions,
we are able to construct a Marchenko equation for elliptic solitons.
A mechanism to obtain nonsingular real solutions from this elliptic DL scheme is formulated.
By utilizing elliptic $N$th roots of unity and reductions,
the elliptic DL schemes, Marchenko equations and  nonsingular real solutions
are studied for the Korteweg-de Vries equation and Boussinesq equation.
Illustrations of the obtained solutions show solitons and their interactions on a periodic background.

\end{abstract}

%

\section{Introduction}\label{sec-1}

The method of direct linearization (DL) was established by Fokas and Ablowitz in 1981 \cite{81FA}.
It is based on singular linear integral equations over arbitrary contours in the complex space of the spectral parameter.
For the Korteweg-de Vries (KdV) equation
\begin{equation}\label{kdv0}
u_t+6uu_x+u_{xxx}=0,
\end{equation}
the linear integral equation in the DL approach is \cite{81FA}
\begin{equation}\label{DL-IE}
\varphi(k;x,t)+\mathrm{i} e^{\mathrm{i}(kx+k^3 t)}\oint_L \frac{\varphi(l;x,t)}{l+k}\mathrm{d}\lambda(l)
=e^{\mathrm{i}(kx+k^3 t)},
\end{equation}
where $\mathrm{i}^2=-1$, $\mathrm{d}\lambda(l)$ and $L$ are appropriate measure and contour.
The plane wave factor (PWF) $e^{\mathrm{i}(kx+k^3 t)}$ is notable
as an element to indicate dispersion relation of the equation under investigation.
For the function $\varphi(k;x,t)$ defined above, it turns out that
$\psi(k;x,t)=\varphi(k;x,t)  e^{-\frac{\mathrm{i}}{2}(kx+k^3 t)}$ and
$u=-\frac{\partial}{\partial x}\oint_L \varphi(l;x,t)\mathrm{d}\lambda(l)$
satisfy the Lax pair of the KdV equation \eqref{kdv0} and
consequently $u$ is a solution of the KdV equation.
Not only solutions, but the Gel'fand-Levitan-Marchenko (GLM) equation in the inverse scattering transform
can be derived from the integral equation \cite{81FA}.
In the framework of \cite{81FA}, solutions of the investigated nonlinear equation are verified
through checking its Lax pair.
Such a scheme was  extended to the  Kadomtsev-Petviashvili (KP) equation,
the Benjamin-Ono equation, etc, to find their solutions, e.g.
\cite{82AF,83-AFA-BO,82FA,83-FA-KP,84SAF,82FVQC,82QNC}.

Soon after Fokas and Ablowitz's pioneer  work \cite{81FA}, the DL idea was developed by
the Dutch group in \cite{83-NQC,83NQVC,84QNCV} into a flexible machinery involving infinite matrix structures,
in which nonlinear equations can be constructed together with their solutions,
without needing to check Lax pairs.
It allows the treatment of discrete as well as continuous equations on one and the same
footing, within one formulism.
Such a DL scheme has shown effectiveness in particular in the investigation of integrable lattice equations,
e.g. \cite{85N,85-NCW,92NPCQ,18F,17FN,21FN,21YF,12ZZN}.

A quite recent progress of the DL approach is the establishment of its elliptic scheme
involving infinite matrix structure and `discrete' Lam\'e-type PWFs \cite{NSZ-2019}.
The Lam\'e function is a solution of the Lam\'e equation
 (known as the Schr\"odinger equation with an elliptic potential)
$ y''+(A+B\wp(x))y=0$
when taking $A=-\wp(k)$ and $B=-2$.
The solution reads $y=\frac{\sigma(x+k)}{\sigma(x)\sigma(k)}e^{-\zeta(k)x}$,
where $\sigma, \zeta$ (and $\wp$) are the Weierstrass functions.
For the KdV equation, there are exact solutions involving the Lam\'e function,
which
 emerged in \cite{74KM} and \cite{74Wahlqulst} 1970s.
In 2010, Nijhoff and Atkinson \cite{IMRN2010} employed a discrete analogue of the Lam\'e function
and constructed solutions for the Adler-Bobenko-Suris \cite{ABS03} (except Q4 equation) lattice equations.
The obtained solutions are presented in terms of the elliptic Cauchy matrix and
called `elliptic solitons' to distinguish from the algebro-geometric (finite-gap) solutions (cf.\cite{DN-JETP-1974}).
Later, their treatment was extended to the lattice KP equations \cite{2013-YN-JMP}.
The elliptic DL scheme established in \cite{NSZ-2019} is based on the infinite matrix structures.
The concept of `elliptic $N$th  root of unity' was introduced to define discrete dispersion relations
as well as make reductions from the lattice KP system to the lower dimensional equations,
namely, the discrete KdV and Boussinesq equations.
As for the elliptic solitons for the continuous equations, a recent progress is due to \cite{22LxZhang},
which  established a framework of bilinear approach that involves the Lam\'e function,\
including formulae of $\tau$ functions, vertex operators and bilinear identities for the KP hierarchy
and their reductions. Periodic degenerations (i.e. trigonometric/hyperbolic and rational forms)
of the elliptic matter were also discussed in \cite{22LxZhang}.
Later, the study of \cite{22LxZhang} was extended to discrete case in \cite{24LxZhang}.
More recently, the consideration of such type of elliptic soliton solutions from
Sato's approach was developed by Kakei \cite{23Kakei},
where the Darboux transformation based on pseudo-differential operators and the Lam\'e function
was established getting elliptic $N$-soliton solutions for the whole KP hierarchy
and the KdV hierarchy.
The elliptic solitons are also related to a recent work by Nakayashiki \cite{Naka-LMP-2024}.

In this paper, we will extend the Fokas-Ablowitz DL scheme to the elliptic soliton case\footnote{
In the paper by elliptic soliton we mean the solutions involving the Lam\'e function
instead of the usual exponential functions for the ordinary solitons,
which follows the naming in \cite{IMRN2010}.
Strictly speaking, ``elliptic solitons'' should be elliptic in the spatial variable $x$ (cf.\cite{K-FAA-1980}).
We will provide a brief review on such ``elliptic solitons'' in Appendix \ref{App-3} }
and formulate a mechanism to obtain non-singular real-valued elliptic soliton solutions.
First, for the KP equation, we will present an integral equation involving elliptic PWFs and Lam\'e function,
using which we can define a wave function and solution for the KP equation.
This DL scheme will be confirmed by checking the wave function and solution satisfy the KP Lax pair.
Real-valuedness of the Weierstrass functions will be discussed (see Appendix \ref{App-1}).
In light of this, we are able to construct a Marchenko equation for elliptic soliton solutions of the KP equation.
By taking special measure in the integral equation, we may obtain explicit expressions
for the elliptic solitons (as well as $\tau$ functions) of the KP equation.
Some solutions will be illustrated.
In principle, employing the notion of elliptic $N$th  root of unity introduced in \cite{NSZ-2019},
we can reduce the elliptic DL scheme, Marchenko equation and solutions  of the KP equation
to those of the KdV equation and the Boussinesq equation.
However, in practice, there are many details to be fixed, which will be elaborated in the paper.

We organize the paper as follows.
In section \ref{sec-2}, we set up the elliptic DL scheme for the KP equation,
construct its Marchenko equation  and derive formulae of the elliptic $N$-soliton solution
and $\tau$-function.
Then, in section \ref{sec-3}, we elaborate reduction details so that
the elliptic DL schemes, Marchenko equations and elliptic solitons
can be constructed for the KdV equation and Boussinesq equation.
Conclusions are given in section \ref{sec-4}.
There are three Appendix sections. The first one is devoted to the real-valuedness of the Weierstrass functions.
The second one consists of the analysis about elliptic cube roots.
The third one includes a brief review on ``elliptic solitons''.

\section{Elliptic DL scheme of the KP equation}\label{sec-2}

In this section, we will focus on the KP equation:
\begin{equation}\label{eq:kp}
u_t+u_{xxx}+6uu_x+3\partial^{-1}_x u_{yy}=0,
\end{equation}
where $\partial^{-1}_x \partial_x=\partial_x \partial^{-1}_x=1$, $\partial_x=\frac{\partial}{\partial x}$.
We will present its elliptic DL scheme and the Marchenko equation
when the equation has a periodic  background solution
\begin{equation}\label{u0}
u_0(x)=-2\wp(x).
\end{equation}
Real-valued solutions and their dynamics will also be presented.

 \subsection{Elliptic DL scheme}\label{sec-2-1}

It is well known that the KP equation \eqref{eq:kp} admits a Lax pair
\begin{subequations}\label{kp-lax}
\begin{align}
	& P\, \psi(x,y,t)  = 0,  ~~~ P\doteq \partial_y-\partial_x^2-u(x,y,t), \label{2.3a} \\
	& M\, \psi(x,y,t) = 0,  ~~~
      M\doteq\partial_t +4\partial_x^3+6u(x,y,t)\partial_x+3u_x(x,y,t)+3\partial_x^{-1}u_y(x,y,t),
\end{align}
\end{subequations}
and the KP equation arises from the compatibility $[P,M]=PM-MP=0$.
When $u(x,y,t)=u_0(x)$, the above Lax pair admits a special solution
\begin{equation}\label{psi-0}
\psi_0(x,y,t;k)= \Psi_x(k)\rho_k(y,t),
\end{equation}
where $k$ is a parameter, $\rho_0(k)$ is a function of $k$ but independent of $(x,y,t)$,
\begin{equation}
\rho_k(y,t)= \rho_0(k)\exp\Bigl(\wp(k)y-2\wp'(k)t\Bigr)\label{rho},
\end{equation}
and $\Psi_x(k)$ is known as the Lam\'e function
\begin{equation}\label{psi-phi}
 \Psi_x(k)=\frac{\sigma(x+k)}{\sigma(x)\sigma(k)}e^{-\zeta(k)x}.
\end{equation}
$\psi_0(x,y,t;k)$  can be viewed as a Lam\'e-type PWF for the
KP equation \cite{22LxZhang}.
Note that among the Weierstrass functions,
$\wp(z)$ is even and doubly periodic, i.e.,
\[\wp(z+2mw+2nw')=\wp(z),~~  m,n\in \mathbb{Z}\]
where $w$ and $w'$ are the  fundamental half-periods,
$\zeta(z)$ and $\sigma(z)$ are odd and quasi-periodic,
 in the sense that
\begin{subequations}\label{periodicity}
\begin{align}
&\zeta(z+2m w+2n w')=\zeta(z)+2m \zeta(w)+2n\zeta(w'), \\
&\sigma(z+2mw+2nw')=(-1)^{m+n+mn}\sigma(z)e^{(z+mw+nw')(2m\zeta(w)+2n\zeta(w'))}. \label{periodicity-sig}
\end{align}
\end{subequations}
When $k\neq w,w'$, $\Psi_x(k)$ and $\Psi_x(-k)$ are linearly independent.
$\Psi_x(k)$ is doubly periodic with respect to $k$:
\begin{equation}
\Psi_x(k+2mw+2nw') =\Psi_x(k). \label{periodicity-Psi}
\end{equation}

In the following we describe the elliptic DL scheme of the Fokas-Ablowitz type
for the KP equation.

\begin{theorem}\label{th-KP-1}
Consider the integral equation of  $\psi(x,y,t;k)$:
\begin{equation}\label{kp-integ}
\psi(x,y,t;k)+\rho_k(y,t)\iint_D\psi(x,y,t;l)\gamma_{l'}(y,t)\Psi_x{(k,l')}\mathrm{d}\lambda(l,l')=\Psi_x(k)\rho_k(y,t),
\end{equation}
where $D$ and $\mathrm{d}\lambda(l,l')$ are general integration domain and measure
in the space of the spectral variables $l$ and $l'$,
$\rho_k(x,y,t)$ and $\Psi_x(k)$ are defined as in \eqref{rho} and \eqref{psi-phi},
$\Psi_x{(k,l')}$ is a generalized Lam\'e function defined as
\begin{equation}\label{psi-2}
\Psi_x{(k,l')}=\frac{\sigma(x+k+l')}{\sigma(x)\sigma(k+l')}e^{-\big(\zeta(k)+\zeta(l')\big)x}
\end{equation}
and $\gamma_{l'}(y,t)$ is defined as
\begin{equation}\label{sigma}
\gamma_{l'}(y,t)=\gamma_0(l')\exp\Big(-\wp(l')y-2\wp'(l')t\Big)
\end{equation}
with $\gamma_0(l')$ being a function of $l'$.
It is assumed that the differential $\partial_\mu$ commutes with $\iint_D$ for $\mu=x, y, t$
and the homogeneous integral equation has only  zero solution.
Under these two assumptions,
solution of the KP  equation \eqref{eq:kp} is provided through
	\begin{equation}\label{u:KP}
	u(x,y,t)=-2\wp(x)-2\partial_x\iint_D \psi(x,y,t;l)\gamma_{l'}(y,t)\Psi_x(l')\mathrm{d}\lambda(l,l').
	\end{equation}
\end{theorem}

\begin{proof}
Thanks to the generalization of direct linearization in \cite{84SAF}, we  apply it to the elliptic case.
We aim to prove that $\psi(x,y,t;k)$ defined by \eqref{kp-integ}
together with $u(x,y,t)$ defined by \eqref{u:KP} satisfy the KP Lax pair \eqref{kp-lax},
i.e. $P\psi=0$ and $M\psi=0$, thereby confirming the theorem.
To achieve that, we consider the following general integral equation
	\begin{equation}\label{eq:integ-kpa}
	\psi(x,y,t;k)+\iint_D\psi(x,y,t;l)h(x,y,t;k,l')\mathrm{d}\lambda(l,l')=\psi_0(x,y,t;k),
	\end{equation}
where $\psi_0(x,y,t;k) $ is the wave function given by  \eqref{psi-0}
and $h(x,y,t;k,l')$  is a function to be determined.
We omit the  independent variables and  spectral parameters in case  where it does not lead to confusion.
Next, we look for $h(x,y,t;k,l')$ such that \eqref{2.3a} holds for $\psi(x,y,t;k)$ defined by \eqref{eq:integ-kpa}.
Applying  the operator $P$  on the above integral equation, it turns out that we can have
\begin{align}\label{P-psi}
P\psi(x,y,t;k)+\iint_D (P\psi(x,y,t;l) )h(x,y,t,k, l') \mathrm{d}\lambda(l,l')=0
\end{align}
if we require
\begin{equation}\label{eq:kp-u-a}
	(u+2\wp(x))\psi_0=\iint_D\Big(2\psi_x h_x-\psi\big(h_y-h_{xx}\big)\Big)\mathrm{d}\lambda(l,l').
\end{equation}
Assume that there exist $f(x,y,t;l')$ and $g(x,y,t;l')$ such that the derivatives of $h(x,y,t,k, l')$
admit the following variable separation forms:
\begin{subequations}\label{kp-hxy}
\begin{align}
& 2 h_x(x,y,t;k,l')=f(x,y,t;l')\psi_{0}(x,y,t;k), \label{kp-hx}\\
& h_{y}(x,y,t;k,l')-h_{xx}(x,y,t;k,l')=-g(x,y,t;l')\psi_{0}(x,y,t;k), \label{kp-hy}
\end{align}
\end{subequations}
which leads \eqref{eq:kp-u-a} to
\begin{equation}\label{2.16}
u(x,y,t)=-2\wp(x)+\iint_D\big(\psi_x(x,y,t;l) f(x,y,t;l')+\psi(x,y,t;l) g(x,y,t;l')\big)\mathrm{d}\lambda(l,l').
\end{equation}
The compatibility $h_{xy}=h_{yx}$ requires that $g=f_x$
and
\begin{equation}\label{kp-lax-adj-x}
f_y(x,y,t;l')=-f_{xx}(x,y,t;l')+2\wp(x)f(x,y,t;l').
\end{equation}
Thus we write \eqref{kp-hy} as
\begin{equation}
2 h_{y}(x,y,t;k,l')=f(x,y,t;l')\psi_{0,x}(x,y,t;k)-f_x(x,y,t;l')\psi_{0}(x,y,t;k). \label{kp-hyb}
\end{equation}
Noticing that \eqref{kp-lax-adj-x} is nothing but an adjoint form of the
linear problem \eqref{2.3a} with $u(x,y,t)=u_0(x)=-2\wp(x)$,
we take (cf.\eqref{psi-0})
\begin{equation}\label{f}
f(x,y,t;l')=-2\gamma_{l'}(y,t)\Psi_x(l')
\end{equation}
where $\gamma_{l'}(y,t)$ is defined as \eqref{sigma}.
Thus, $h(x,y,t;k,l')$ is well defined via \eqref{kp-hx} and \eqref{kp-hyb} with $\psi_{0}(x,y,t;k)$
and the above $f(x,y,t;l')$.
It is not difficult to get $h(x,y,t;k,l')$ by integrating \eqref{kp-hx} and \eqref{kp-hyb}.
Thus, we have
\begin{equation}\label{h}
h(x,y,t;k,l')
=\frac{1}{2(\wp(k)-\wp(l'))}\left(f(x,y,t;l')\psi_{0,x}(x,y,t;k)-f_x(x,y,t;l')\psi_0(x,y,t;k)\right),
\end{equation}
which can be written as
\begin{equation}\label{kp-h}
h(x,y,t;k,l')=\gamma_{l'}(y,t)\rho_k(y,t)\Psi_x(k, l'),
\end{equation}
where the following formulas have been used:
\begin{subequations}\label{Weieratrass-1}
\begin{align}
&\wp(z)-\wp(u)=-\frac{\sigma(z+u)\sigma(z-u)}{\sigma^2(z)\sigma^2(u)},\\
&\zeta(u)+\zeta(v)+\zeta(z)-\zeta(u+v+z)
=\frac{\sigma(u+v)\sigma(u+z)\sigma(z+v)}{\sigma(u)\sigma(v)\sigma(z)\sigma(z+u+v)}.
\label{eq:indent-chi}
\end{align}
\end{subequations}
Substituting \eqref{kp-h} into \eqref{eq:integ-kpa} and substituting \eqref{f} together with $g=f_x$
into \eqref{2.16}, we can recover the integral equation \eqref{kp-integ} and the solution \eqref{u:KP}, respectively.
Meanwhile, with \eqref{kp-h}, the homogeneous integral equation \eqref{P-psi}
has   only zero solution $P \psi(x,y,t;k)=0$.

Now, what remains is to show $M \psi(x,y,t;k)=0$ in light of the integral equation \eqref{kp-integ}
and the solution \eqref{u:KP}.
First, note that the expression \eqref{h} of $h(x,y,t;k,l')$ gives
\begin{equation}\label{ht}
h_t(x,y,t;k,l')=2\psi_{0,x}(x,y,t;k)f_{x}(x,y,t;l')-2\Big(\wp(k)+\wp(l')+\wp(x)\Big)\psi_0(x,y,t;k)f(x,y,t;l').
\end{equation}
Applying the operator $M$ on  \eqref{kp-integ} yields
\begin{equation*}
\begin{split}
&M\psi+\iint_D(M\psi)h\mathrm{d}\lambda(l,l')
+\iint_D\psi h_t\mathrm{d}\lambda(l,l')\\
& +12\iint_D(\psi_xh_{xx}+\psi_{xx} h_x)\mathrm{d}\lambda(l,l')
 + \iint_D (4\psi h_{xxx}+6u\psi h_x)\mathrm{d}\lambda(l,l')\\
=~& 3\Big((u+2\wp(x))_x+\partial^{-1}u_y\Big)\psi_0+6(u+2\wp(x))\psi_{0,x}.
\end{split}
\end{equation*}
By inserting $h_t(x,y,t;k,l')$ \eqref{ht} and $u(x,y,t)$ \eqref{u:KP},
it reduces to a homogeneous integral equation
\begin{align*}
M\psi(x,y,t;k)+\rho_k(y,t)\iint_D (M\psi(x,y,t;l) ) \gamma_{l'}(y,t)\Psi_x(k, l') \mathrm{d}\lambda(l,l')=0,
\end{align*}
which implies that $M\psi(x,y,t;k)=0$
in light of the uniqueness assumption of solutions of the integral equation.

Thus, we complete the proof.

\end{proof}

%

\subsection{Connection with the Marchenko equation}\label{sec-2-2}

From the linear integral equation \eqref{kp-integ}, we can derive the Marchenko integral equation
for elliptic solitons (cf.\cite{74ZS,77MZBIM} for the usual soliton case), which involves integration on the real axis.
Before we proceed, we first refer the reader to Appendix \ref{App-1} of the paper for
the real-valuedness of the Weierstrass functions.

In the following we consider the case where the moduli $g_2$ and $g_3$ are real and
the discriminant $\Delta$ is positive. In this case, according to Proposition \ref{prop-A.2},
the half-periods $w$ and  $w'$ are real and  purely imaginary, respectively.
We also suppose $\mathrm{Im} w'>0$ without loss of generality.

Note that  $\sigma(z)$  takes zero when $z=2mw+2nw'$, $m,n\in\mathbb{Z}$,
and even when $z\in \mathbb{R}$  there are periodic zeros when $z=2mw$ for $m\in\mathbb{Z}$,
which will bring singularities to elliptic solitons.
In order to get real-valued non-singular solutions,
we shift the zeros from the real axis to the  complex plane by performing a transformation $x\mapsto x+w'$.
Thus we introduce
\begin{subequations}\label{2.26}
 \begin{align}
\widetilde{ \Psi}_x(k)&\doteq \frac{\sigma(x+k+w')}{\sigma(x+w')\sigma(k)}e^{-\zeta(k)x-\zeta(w')k},
\label{Ps-k}\\
\widetilde{ \Psi}_x(k,l)&\doteq \frac{\sigma(x+k+l+w')}{\sigma(x+w')\sigma(k+l)}
e^{-(\zeta(k)+\zeta(l))x-\zeta(w')(k+l)}.
\label{Ps-kl}
\end{align}
\end{subequations}

\begin{Proposition}\label{prop-2.1}
$\widetilde{ \Psi}_x(k)$ and $\widetilde{ \Psi}_x(k,l)$ are real-valued functions
when $x, k, l, g_2, g_3\in \mathbb{R}$, and $\Delta>0$.
\end{Proposition}

\begin{proof}
Using Proposition \ref{prop-A.1}, \ref{prop-A.2} and the trick used in the proof Proposition \ref{prop-A.1},
together with the quasi-periodicity \eqref{periodicity-sig} of $\sigma(z)$,
one can prove the proposition.

\end{proof}

The function $\widetilde{ \Psi}_x(k)$ has been well studied in Chapter 5 of \cite{book-Pastras}.
For the function $\widetilde{ \Psi}_x(k,l)$, we rewrite it in the form
\begin{subequations}
\begin{equation}
\widetilde{ \Psi}_x(k,l) =v(x;k,l)e^{\frac{p(k)+p(l)}{w}x}
\end{equation}
where
\begin{align}
& v(x;a,b)=\frac{\sigma(x+w'+a+b)}{\sigma(x+w')\sigma(a+b)}
e^{-\zeta(w')(a+b)-\frac{(a+b)\zeta(w)}{w}x},\\
& p(a)=a\zeta(w)-w\zeta(a).
\end{align}
\end{subequations}
Then we have the following.

\begin{Proposition}\label{prop-2.2}
The function $\widetilde{ \Psi}_x(k,l)$ in the above form has the following properties.\\
(1). (Quasi-)periodicity:
\begin{subequations}
 \begin{align}
 &\widetilde{ \Psi}_x(k+2 w,l)=-\widetilde{ \Psi}_x(k,l), \label{3.28a} \\
 & \widetilde{ \Psi}_x(k+2 w',l)=\widetilde{ \Psi}_x(k,l). \label{3.28b}
 \end{align}
\end{subequations}
(2). For function $v(x; k,l)$ we have
\begin{equation}\label{2.29}
 v(x+2w;k,l)=v(x;k,l).
\end{equation}
(3). For  function $p(k)$ we have
\begin{align}
 p(k+2w)=p(k), \quad p(k+2w')=p(k)+\pi \mathrm{i}, \quad p(-k)=-p(k).
\end{align}
 When $k\in(0, w']$ and $(w,w+w']$, $p(k)$ is  purely imaginary;
 when $k\in(w', w+w')$,
 $p(k)$ is  complex, and the imaginary part is always $\frac{\pi}{2}$; when $k\in(0, w]$,
 $p(k)$ is  real, it varies monotonically from $-\infty$ at the origin to $0$ at $w$.

Note that in the above $(c,d)$ (or $[c,d]$ etc.) for $c,d\in \mathbb{C}$ means the segment
on the complex plane from point $c$ to $d$, rather than an interval on the real axis.
\end{Proposition}

\begin{proof}
(1). The proof of \eqref{3.28a} employs  formula $w'\zeta(w)-w\zeta(w')=\frac{\pi \mathrm{i} }{2}$
and the quasi-periodicity properties \eqref{periodicity}.
The proof of \eqref{3.28b} is straight forward.\\
(2). \eqref{2.29} can be verified directly.\\
(3). The proof can be found on pages 61 and 62 of \cite{book-Pastras}.

\end{proof}

Next, we construct the Marchenko equation for elliptic solitons.
Considering the KP equation \eqref{eq:kp} allows a shift symmetry $u(x)=u(x+c)$
and the parameters $\rho_0(k)$ in \eqref{rho} and $\gamma_0(l')$ in \eqref{sigma} are arbitrary,
we can immediately convert the DL scheme in Theorem \ref{th-KP-1} to the following form.

\begin{theorem}\label{th-KP-1'}
Assume $\rho_k(y,t)$, $\gamma_{l'}(y,t)$, $\widetilde{ \Psi}_x(k)$ and $\widetilde{ \Psi}_x(k,l)$ are defined as
in \eqref{rho}, \eqref{sigma} \eqref{Ps-k} and \eqref{Ps-kl}.
For $\psi(x,y,t;k)$ defined by the integral equation
\begin{equation}\label{integ-KP-move}
\psi(x,y,t;k)+\rho_k(y,t) \iint_D   \psi(x,y,t;l)\gamma_{l'}(y,t)\widetilde{ \Psi}_{x}{(k,l')}\mathrm{d}\lambda(l,l')
=\widetilde{ \Psi}_{x}(k)\rho_k(y,t),
\end{equation}
the KP equation \eqref{eq:kp} allows a solution
\begin{equation}
u(x,y,t)=-2\wp(x+w')-2\partial_x\iint_D \psi(x,y,t;l)\gamma_{l'}(y,t)
\widetilde{ \Psi}_{x}(l')\mathrm{d}\lambda(l,l').\label{u:KP-1}
\end{equation}
\end{theorem}

The integral equation \eqref{integ-KP-move} is ready for deriving the Marchenko equation.
We multiply it by $\widetilde{ \Psi}_{s}(k') \gamma_{k'}(y,t)$
and integrate each term on the domain $D$ in the space of variables $k$  and $k'$.
For the first term on the left-hand side and the term on the right-hand-side, we denote them as
\begin{subequations}
\begin{align}
K(x,s,y,t)&=-\iint_{D}\psi(x,y,t;l)\widetilde{ \Psi}_{s}(l')\gamma_{l'}(y,t)\mathrm{d}\lambda(l,l'),\\
F(x,s,y,t)&=\iint_{D}\widetilde{ \Psi}_{x}(k)\rho_k(y,t)\widetilde{ \Psi}_{s}(k')\gamma_{k'}(y,t)
\mathrm{d}\lambda(k,k').
\end{align}
\end{subequations}
The remaining term reads
\begin{equation}\label{2.34}
\iint_{D}\mathrm{d}\lambda(k,k')   \iint_{D}  \mathrm{d}\lambda(l,l') \widetilde{\Psi}_{s}(k')
\gamma_{k'}(y,t)\psi(x,y,t;l)\rho_k(y,t)\widetilde{ \Psi}_{x}(k,l')\gamma_{l'}(y,t).
\end{equation}
Note that by making use of the formula \eqref{eq:indent-chi},  we have a relation
\begin{equation}\label{kp:inte-cond}
\frac{\mathrm{d}}{\mathrm{d}x}\widetilde{ \Psi}_x(k,l')=-\widetilde{ \Psi}_{x}(k)\widetilde{ \Psi}_{x}(l').
\end{equation}
In what follows, we restrict ourselves to the case where $x, g_2, g_3\in \mathbb{R}$, $w>0$,
$\mathrm{Im} w' >0$ and $\Delta >0$.
Particularly, we take $k,l'\in (0,w)$.
It then follows from item (2) and (3) of Proposition \ref{prop-2.2} that
$v(x;k,l')$ is bounded on the real axis, and  $p(k)+p(l')<0$,
and consequently,  $\widetilde{ \Psi}_x(k,l')$  exponentially decreases to zero
as $x\to +\infty$.
Thus, from \eqref{kp:inte-cond} we have
\begin{equation}\label{eq:kp-GLM-condit}
\int_x^{+\infty}\widetilde{ \Psi}_{\xi}(k)\widetilde{ \Psi}_{\xi}(l')\mathrm{d}\xi
=\widetilde{ \Psi}_x(k,l'),
\end{equation}
which leads \eqref{2.34} to a form
\begin{equation}
-\int_x^{+\infty}K(x,\xi,y,t)F(\xi,s,y,t)\mathrm{d}\xi.
\end{equation}
Finally, we can conclude the following.

\begin{theorem}\label{th-KP-2}
There exist $K(x,s,y,t)$ and $F(x,s,y,t)$ satisfying the  Marchenko equation
\begin{equation}\label{eq:GLM-KP}
K(x,s,y,t)+F(x,s,y,t)+\int_x^{+\infty}K(x,\xi,y,t)F(\xi,s,y,t)\mathrm{d}\xi=0,
\end{equation}
and  the solution of the KP equation \eqref{eq:kp} is given by
\begin{equation}
u(x,y,t)=-2\wp(x+w')+2\partial_x K(x,x,y,t).
\end{equation}
\end{theorem}

\subsection{ Elliptic soliton solutions}\label{sec-2-3}

In the following, we construct elliptic soliton  solutions of the KP equation from the integral equation \eqref{kp-integ}
in Theorem \ref{th-KP-1}.
We impose the following condition on the measure (cf.\cite{18-FuweiNijhoff,NSZ-2019}):
\begin{equation*}
\mathrm{d}\lambda(l,l')=\sum_{i=1}^N\sum_{j=1}^{N'}B_{j,i}\delta(l-k_i)\delta(l'-k'_j)\mathrm{d}l~\mathrm{d}l',
\end{equation*}
where $B_{i,j}\in \mathbb{C}$, $k_i$ and $k'_j$  belong to the fundamental periodic parallelogram
formed by $2w$ and $2w'$, and we assume $w>0, \mathrm{Im}w'>0$.
Then the integral equation \eqref{kp-integ} reduces to  a linear algebraic system
\begin{align}
\psi(x,y,t;k_i)+\rho_{k_i}(y,t)\sum_{i=1}^N\sum_{j=1}^{N'}B_{j,i}\psi(x,y,t;k_i)\gamma_{k'_j}(y,t)\Psi_x(k_i,k'_j)
=\rho_{k_i}(y,t)\Psi_x(k_i) \label{kp:linear-eqs}
\end{align}
and the solution \eqref{u:KP} yields
\begin{align}\label{u:KP-2}
u(x,y,t)=-2\wp(x)-2\partial_x\sum_{i=1}^N\sum_{j=1}^{N'}B_{j,i}\psi(x,y,t;k_i)\gamma_{k'_j}(y,t)\Psi_x(k'_j).
\end{align}
Introduce
\begin{subequations}
\begin{align}
 &\mf{u}^{(0)}=(\psi(x,y,t;k_1), ~\psi(x,y,t;k_2),~\cdots,~\psi(x,y,t;k_N))^{T},\\
& \mathbf{r}=(\rho_{k_1}(y,t)\Psi_x(k_1),~ \rho_{k_1}(y,t)\Psi_x(k_2), ~\cdots,~ \rho_{k_N}(y,t)\Psi_x(k_N))^T,\\
& \mathbf{s}= (\gamma_{k'_1}(y,t)\Psi_x(k'_1),~\gamma_{k'_2}(y,t)\Psi_x(k'_2),~\cdots,~
\gamma_{k'_{N'}}(y,t)\Psi_x(k'_{N'}) )^T.
\end{align}
\end{subequations}
Then, \eqref{kp:linear-eqs} and \eqref{u:KP-2} can be formulated in matrix form
\begin{align}\label{u00}
(\mathbf{I+MB})\mf{u}^{(0)}=\mathbf{r}
\end{align}
and
\begin{equation}\label{u:KP-3}
u(x,y,t)=-2\wp(x)-2\partial_x( \mf{s}^T\mf{B}\mf{u}^{(0)}),
\end{equation}
where $\mf{I}$ is a $N\times N$ unit matrix,
 $\mathbf{B}=(B_{j,i})_{N'\times N}$ is an arbitrary constant matrix,
and $\mathbf{M}$ is a (dressed) elliptic Cauchy matrix
with each element being doubly periodic with respect to $k_i$ and $k'_j$, defined as
\begin{align}
 \mathbf{M}=(M_{i,j})_{N\times N'},\quad M_{i,j}=\rho_{k_i}(y,t)\Psi_x(k_i,k'_j)\gamma_{k'_j}(y,t).\label{M-KP}
\end{align}
Solving for $\mf{u}^{(0)}$ from \eqref{u00} and inserting it into \eqref{u:KP-3}, we arrive at
\begin{equation}\label{u:KP-4}
 u(x,y,t)=-2\wp(x)-2\partial_x (\mf{s}^T\mf{B}(\mf{I+MB})^{-1}\mf{r}),
\end{equation}
which presents an elliptic multi-soliton solution for the KP equation \eqref{eq:kp}.

We can construct the $\tau$-function by defining
\begin{equation}
\tau=|\mathbf{I+MB}|.
\end{equation}
Using a trace formula (see \cite{14-XuZhangZhao})
\begin{equation*}
\mf{s}^T\mf{B}(\mf{I+MB})^{-1}\mf{r}=\mathrm{Tr}\left(\mf{r}\mf{s}^T\mf{B}(\mf{I+MB})^{-1}\right)
= \mathrm{Tr}\left(-(\mf{I+MB})_x(\mf{I+MB})^{-1}\right)=-\frac{\tau_x}{\tau}
\end{equation*}
where relation $\mf{M}_x=-\mf{r}\mf{s}^{T}$ has been used,
solution \eqref{u:KP-4} is expressed as
\begin{equation}
u(x,y,t)=2\partial_x^2\ln\big(\sigma(x)\tau\big). 
\end{equation}
In the following, for convenience, we consider a special case where we take $N=N'$ and $\mf{B}=\mf{I}$.
The $\tau$ function  can be expanded into Hirota's form using the expansion
\begin{equation}
|\mathbf{I+M}|=1+\sum_{i=1}^N|M_{i,i}|+\sum_{i<j}\left|\begin{array}{cccc}
M_{i,i} & M_{i,j}  \\
M_{j,i}  &  M_{j,j}
\end{array}\right|+\cdots+|\mathbf{M}|,
\end{equation}
and each  term  can be expanded by the Frobenius formula\cite{82Frobenius}
\begin{equation*}
\begin{split}
\Big|\rho_{k_i} \Psi_{x}(k_i,l_j)\gamma_{l_j}\Big|_{1\leq i,j\leq N}\!\!\!
=\!\! \left(\!\prod_i^N\rho_{k_i}\gamma_{l_i}e^{-(\zeta(k_i)+\zeta(l_i))x}\!\right)\!\!
\frac{\sigma(x+\sum_{i=1}^N(k_i+l_i))}{\sigma(x)\prod_{i=1}^N\sigma(k_i+l_i)}\!\!
\prod_{i<j}\frac{\sigma(k_i-k_j)\sigma(l_i-l_j)}{\sigma(k_i+l_j)\sigma(l_i+k_j)}.
\end{split}
\end{equation*}
Consequently, we obtain an explicit form for $\tau$ function (also see \cite{23Kakei,22LxZhang}):
\begin{equation}
\tau=\sum_{J\subset S}\left(\prod_{i\in J}c_i\right)
\left(\mathop{\rm{\prod}}_{i<j \in J}A_{i,j}\right)
\frac{\sigma(x+\sum_{i\in J}(k_i+k'_i))}{\sigma(x)\prod_{i\in J}\sigma(k_i+k'_i)}
\mathrm{exp}\left(\sum_{i\in J} \theta_{[e]}(x,y,t;k_i,k'_i)\right),
\end{equation}
where  $c_i=\rho_0(k_i)\gamma_0(k_i')\in \mathbb{C}$,  the sum over all subsets $J$  of $S=\{1,2,\cdots,N\}$, and
\begin{subequations}\label{theta-Aij}
\begin{align}
&  \theta_{[e]}(x,y,t;k_i,k'_i)=-(\zeta(k_i)+\zeta(k'_i)) x+(\wp(k_i)-\wp(k'_i)) y -2(\wp'(k_i)+\wp'(k'_i))t, \\
& A_{i,j} =\frac{\sigma(k_i-k_j)\sigma(k'_i-k'_j)}{\sigma(k_i-k'_j)\sigma(k'_i-k_j)}.\label{A-ij-KP}
\end{align}
\end{subequations}
If we start from Theorem \ref{th-KP-1'}, by the same procedure, we will get a non-singular real-valued solution
\begin{equation}\label{u:KP-5}
u(x,y,t)=2\partial_x^2\ln\big(\sigma(x+w')\tau\big),
\end{equation}
where the $\tau$-function reads ($\tau=\tau_N$)
\begin{equation}
\tau_N=\sum_{J\subset S}\left(\prod_{i\in J}c_i\right)\!\!
\left(\mathop{\rm{\prod}}_{i<j \in J}A_{i,j}\right)\!
\frac{\sigma(x+w'+\sum_{i\in J}(k_i+k'_i))}{\sigma(x+w')\prod_{i\in J}\sigma(k_i+k'_i)}
\mathrm{exp}\left(\sum_{i\in J} \theta_{[e]}(x,y,t;k_i,k'_i)\right),
\end{equation}
and $\theta_{[e]}$ and $A_{i,j}$ are defined as in \eqref{theta-Aij}, but $c_{i}$ should be redefined by
\begin{equation}
c_i=\widetilde{c}_{i}\mathrm{\exp}\left(-\zeta(w')(k_i+k'_i)\right),~~~ \widetilde{c}_{i}\in \mathbb{R}.
\end{equation}

The simplest two solutions are
 the elliptic one-soliton solution (1SS)
\begin{subequations}\label{1ss}
\begin{equation}
u=2 \partial_x^2\ln\big(\sigma(x+w') \tau_1\big)
\end{equation}
where
\begin{equation}
\tau_1=1+c_{1}\frac{\sigma(x+k_1+k'_1+w')}{\sigma(x+w')\sigma(k_1+k'_1)}e^{\theta_{[e]}(x,y,t;k_1,k'_1)},
\end{equation}
\end{subequations}
and  elliptic two-soliton solution (2SS)
\begin{subequations}\label{2ss}
\begin{equation}
u=2\partial_x^2\ln\big( \sigma(x+w')\tau_2\big)
\end{equation}
where
\begin{align}
\tau_2=& 1+c_{1}\frac{\sigma(x+k_1+k'_1+w')}{\sigma(x+w')\sigma(k_1+k'_1)}
e^{\theta_{[e]}(x,y,t;k_1,k'_1)}
+c_{2}\frac{\sigma(x+k_2+k'_2+w')}{\sigma(x+w')\sigma(k_2+k'_2)}e^{\theta_{[e]}(x,y,t;k_2,k'_2)}
\nonumber\\
&
+c_{1}c_{2}A_{1,2}\frac{\sigma(x+k_1+k'_1+k_2+k'_2+w')}{\sigma(x+w')\sigma(k_1+k'_1)\sigma(k_2+k'_2)}
e^{\theta_{[e]}(x,y,t;k_1,k'_1)+\theta_{[e]}(x,y,t;k_2,k'_2)}.
\end{align}
\end{subequations}

To illustrate these solutions, we take $g_2=4$ and $g_3=1$.
The corresponding half-periods can be obtained from elliptic integrals,  numerically given by
\begin{equation}\label{w12}
  w=1.22569..., \qquad   w'=\mathrm{i}\, 1.49673... .
\end{equation}
Then $u(x,y,t)$ is a real function when $ x,y,t,k_j,l_j, \widetilde{c}_{j} \in\mathbb{R}$.
Dynamics of these solutions are  depicted in Fig.\ref{fig:1},
which describe line solitons on a periodic background that is due to $-2\wp(x+w')$.

\begin{figure}[!h]
\centering
\subfigure[]{
\begin{minipage}{5.5cm}
\includegraphics[width=\textwidth]{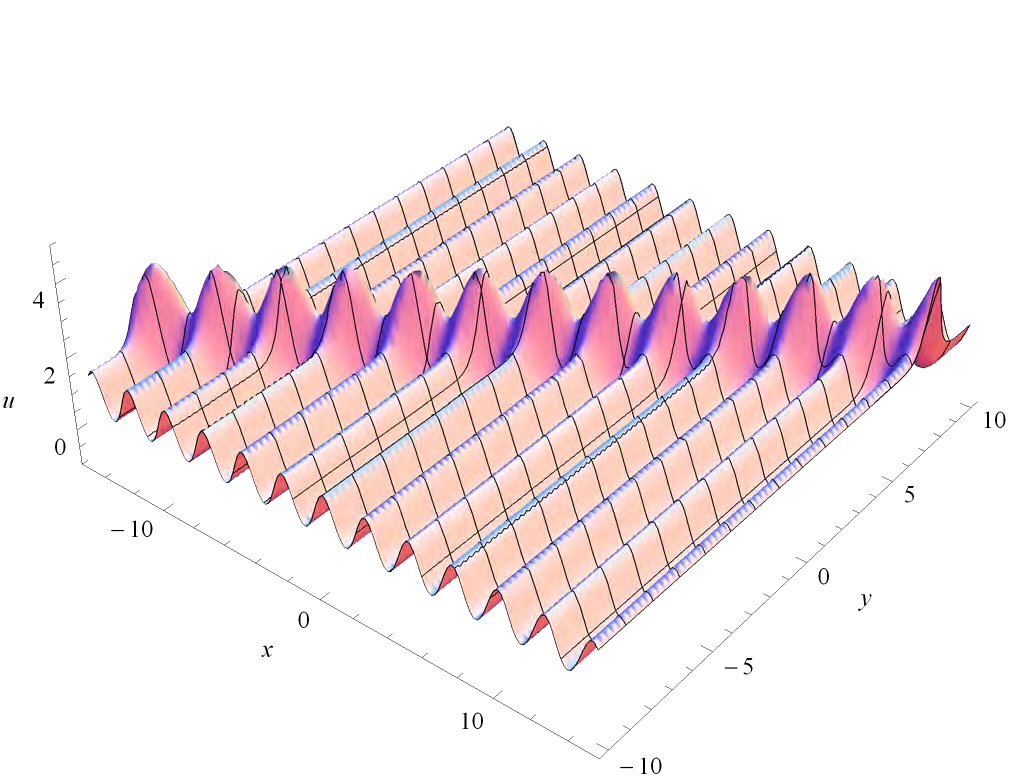}
\end{minipage}}
\hspace{5mm}
\subfigure[]{
\begin{minipage}{5.5cm}
\includegraphics[width=\textwidth]{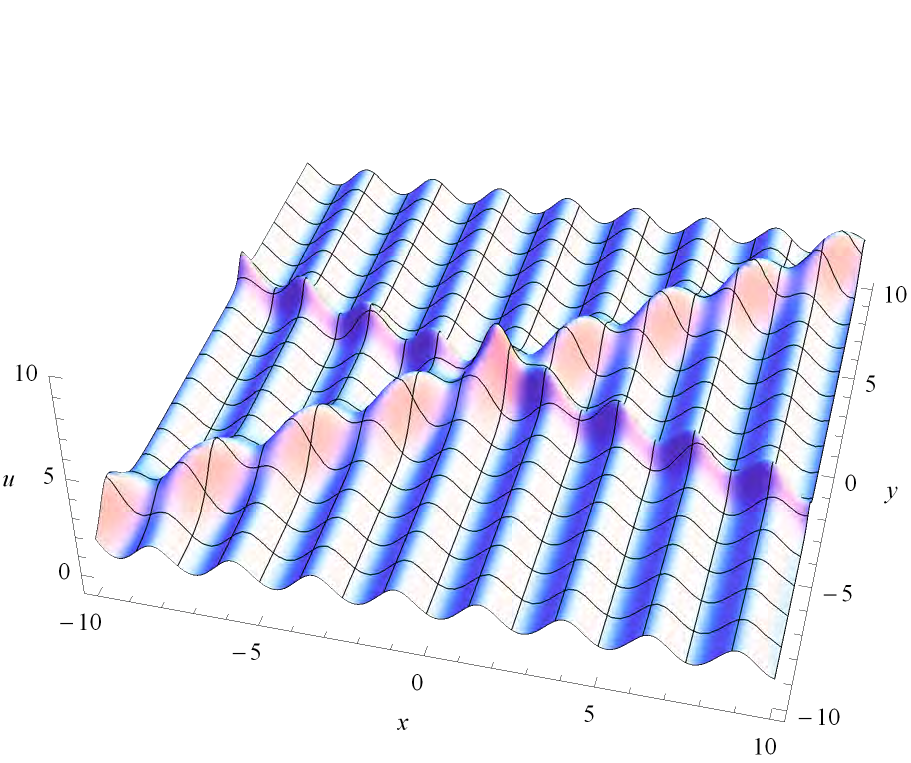}
\end{minipage}}
\caption{Shape and motion of the  elliptic solitons.
(a) Elliptic 1SS given by \eqref{1ss} for $t=0$, $k_1=0.45$, $k'_1=1.25$, $\widetilde{c}_{1}=1$.
(b) Elliptic 2SS given by \eqref{2ss} for $t=0$, $k_1=1.9$, $k'_1=1.0$, $k_2=0.8$, $k'_2=0.4$,
$\widetilde{c}_{1}=\widetilde{c}_{2}=1$.}
\label{fig:1}
\end{figure}

There can be resonance of two line solitons when we take $k'_1=k'_2$ but $k_1\neq k_2$,
such that $A_{1,2}=0$ and $\tau_2$ reads
\begin{equation}\label{Yss}
\tau_2= 1+c_{1}\frac{\sigma(x+k_1+k'_1+w')}{\sigma(x+w')\sigma(k_1+k'_1)}e^{\theta_{[e]}(x,y,t;k_1,k'_1)}
+c_{2}\frac{\sigma(x+k_2+k'_1+w')}{\sigma(x+w')\sigma(k_2+k'_1)}e^{\theta_{[e]}(x,y,t;k_2,k'_1)}.
\end{equation}
In this case, we get Y-shape interactions, as depicted in Fig.\ref{fig:2}.

\begin{figure}[!h]
\centering
\begin{minipage}{5.5cm}
\includegraphics[width=\textwidth]{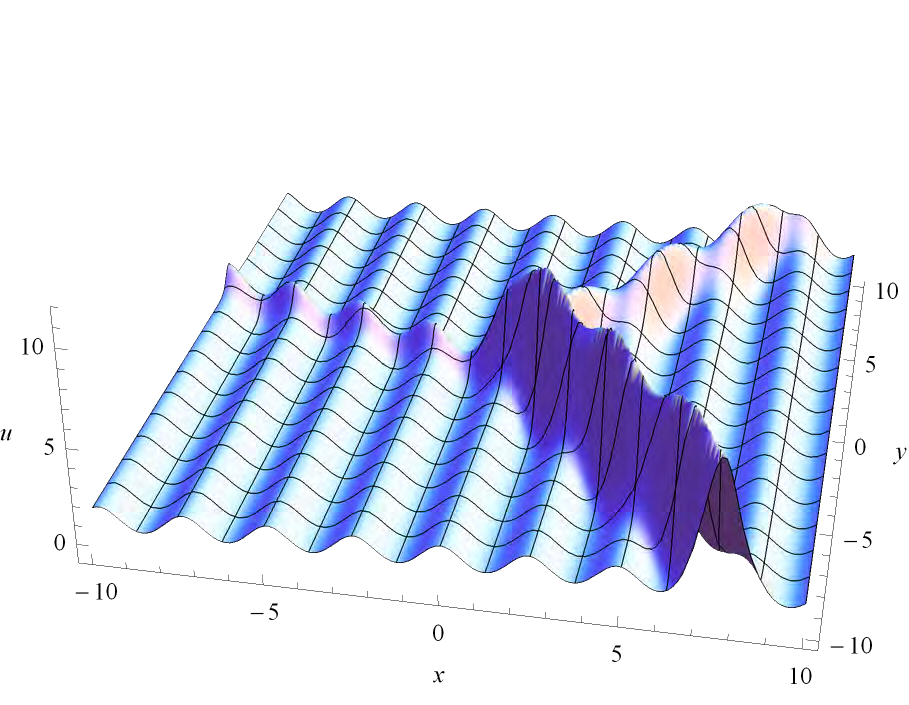}
\end{minipage}
\caption{Shape and motion of  elliptic 2SS resonance with \eqref{Yss} for $t=0$, $k_1=1.9$,
$k'_1=k'_2=0.4$, $k_2=0.8$,
$\widetilde{c}_{1}=\widetilde{c}_{2}=1$.}
\label{fig:2}
\end{figure}

\section{Reduction to the KdV and Boussinesq}\label{sec-3}

\subsection{Elliptic $N$th roots of unity}\label{sec-3-1}

It is well known that soliton solutions of the KP equation can be reduced to those of the KdV equation,
Boussinesq  equation and higher order equations in the Gel'fand-Dickey hierarchy.
This can be done by imposing certain constraints on the eigenvalues (i.e. $k_i$ and $l_i$)
of the KP solution so that the dispersion relations of the lower dimensional equations are obtained.
In soliton case, such reductions are  achieved  through the constraint $l=\omega k$
where $\omega$ is  $N$th  root of unity.
In elliptic soliton case, the elliptic $N$th  roots of unity were introduced in \cite{NSZ-2019},
which is used to reduce the elliptic DL schemes of the lattice KP system to the DL scheme of the lattice KdV system
and lattice Boussinesq  system \cite{NSZ-2019},
as well as to reduce $\tau$-functions and vertex operators in the continuous case \cite{22LxZhang}.
Elliptic $N$th  roots of unity are defined as the following.

\begin{definition}\label{D-1}
\cite{NSZ-2019} There exist distinct $\{\omega_j(\delta)\}_{j=0}^{N-1}$,
up to the periodicity of the periodic lattice, such that the following equation holds,
\begin{subequations}
\begin{equation}\label{root}
\prod^{N-1}_{j=0}\frac{\sigma(\kappa+\omega_j(\delta))}{\sigma(\kappa)\sigma(\omega_j(\delta))}
=\frac{1}{(N-1)!}(\wp^{(N-2)}(-\kappa)-\wp^{(N-2)}(\delta)),
\end{equation}
where $\omega_0(\delta)=\delta$ and all $\{\omega_j(\delta)\}$ are independent of $\kappa$.
$\{\omega_j(\delta)\}_{j=0}^{N-1}$ are called elliptic $N$th  roots of the unity.
These roots also satisfy (modulo the period lattice)
\begin{equation}
\sum^{N-1}_{j=0}\omega_j(\delta)=0, \quad \sum^{N-1}_{j=0}\zeta^{(l)}(\omega_j(\delta))=0,~
l =0,1,\cdots,N-2.
\end{equation}
\end{subequations}
\end{definition}

\subsection{Reduction to the KdV equation}\label{sec-3-2}

\subsubsection{DL scheme}\label{sec-3-2-1}

The integral equation \eqref{kp-integ} of the KP equation can be reduced to the case of the KdV equation
by taking a specific measure
\begin{equation}\label{reduct-kdv}
\mathrm{d}\lambda(l,l')=\delta(l'+\omega_1(l))\mathrm{d}\lambda(l)\mathrm{d}l',
\end{equation}
in which we  use the elliptic square roots  $\omega_0(l)=l$, and $\omega_1(l)=-l$.
The trick to get the integral equation for the KdV equation is the following.
In the first step, using \eqref{reduct-kdv} we reduce the integral equation \eqref{kp-integ} from
a double-integral form to a single-integral form
\begin{equation}\label{3.3}
\psi(x,y,t;k)+\rho_k(y,t) \int_L\psi(x,y,t;l)\gamma_{l}(y,t)\Psi_x{(k,l)}\mathrm{d}\lambda(l)=\Psi_x{(k)} \rho_k(y,t),
\end{equation}
where $L$, $\mathrm{d}\lambda(l)$ are suitable contour and measure.
Next, we introduce $\phi(x,y,t;k)$ such that
\begin{equation}
\psi(x,y,t;k)=\phi(x,y,t;k)\rho_k(y,t),
\end{equation}
by which it follows from \eqref{3.3} that
\begin{equation}\label{3.5}
\phi(x,y,t;k)+\int_L\phi(x,y,t;l)\rho_l(y,t)\gamma_{l}(y,t)\Psi_x{(k,l)}\mathrm{d}\lambda(l)=\Psi_x{(k)}.
\end{equation}
It turns out that in the kernel
\[\rho_l(y,t)\gamma_{l}(y,t)\Psi_x{(k,l)}=\rho_0(l)\gamma_0(l)\Psi_x{(k,l)}e^{-4\wp'(l)t}\]
is independent of $y$,
which means there exists a solution $\phi(x,y,t;k)$ that is independent of $y$ as well.
We denote the above $\phi(x,y,t;k)$ to be $\phi(x,t;k)$ without making any confusion.
Then we introduce
\begin{equation}\label{varrho}
\varrho_k(t)=\varrho_0(k)e^{-2\wp'(k)t},
\end{equation}
and  $\varphi(x,t;k)=\phi(x,t;k) \varrho_k(t)$
and take $\gamma_0(k)=\varrho_0(k)$,
we arrive at an integral equation  from \eqref{3.5}
\begin{equation}\label{3.7}
\varphi(x,t;k)+\varrho_k(t)\int_L\varphi(x,t;l)\varrho_l(t)\Psi_x{(k,l)}\mathrm{d}\lambda(l)=\Psi_x{(k)}\varrho_k(t).
\end{equation}
Repeating the same procedure,  the solution can be reduced as
\begin{equation}\label{u-int-kdv}
u(x,t)=-2\wp(x)-2\partial_x\int_L  \varphi(x,t;l)\varrho_{l}(t)\Psi_x(l)  \mathrm{d}\lambda(l).
\end{equation}
These two settings consist of the elliptic DL scheme for the KdV equation \eqref{kdv0}.

\begin{theorem}\label{th:kdvDLA}
If the integral equation \eqref{3.7} satisfies:
(a) the contour $L$ and measure $\mathrm{d}\lambda(l)$ permit the interchange of $\int_L$
and the differentials w.r.t. $x, t$,
(b) the homogeneous integral equation has only the zero solution,
then the function $u(x,t)$ defined by \eqref{u-int-kdv} solves the KdV equation  \eqref{kdv0}.
\end{theorem}

The proof is similar to the KP case. One can prove that $\varphi(x,t;k)$ and $u(x,t)$ defined by
\eqref{3.7} and \eqref{u-int-kdv} agree with the  Lax pair of the KdV equation:
\begin{subequations}
\begin{align}
&P \varphi(x,t;k)=\Big(\partial_x^2+u(x,t)-\wp(k)\Big)\varphi(x,t;k)=0,\\
&M \varphi(x,t;k)=\Big(-\partial_t+u_x(x,t)-(4\wp(k)+2u(x,t))\partial_x\Big)\varphi(x,t;k)=0.
\end{align}
\end{subequations}
We skip presenting the proof.

\subsubsection{Marchenko equation}\label{sec-3-2-2}

Similar to the treatment in section \ref{sec-2-2},
here we again restrict our discussion to the case where $x, g_2, g_3\in \mathbb{R}$, $w>0$,
$\mathrm{Im} w' >0$ and $\Delta >0$.
These settings allow us to have an alternative elliptic DL scheme which can yield real non-singular
solutions for the KdV equation.

\begin{theorem}\label{th:kdvDLA'}
Assume $\varrho_k(t)$, $\widetilde{ \Psi}_x(k)$ and $\widetilde{ \Psi}_x(k,l)$ are defined as
in \eqref{varrho}, \eqref{Ps-k} and \eqref{Ps-kl}.
For $\varphi(x,t;k)$ defined by the integral equation
\begin{equation}
\varphi(x,t;k)+\varrho_k(t) \int_L   \varphi(x,t;l)\varrho_{l}(t)\widetilde{ \Psi}_{x}{(k,l)}\mathrm{d}\lambda(l)
=\widetilde{ \Psi}_{x}(k)\varrho_k(t),\label{integ-KdV-move}
\end{equation}
the KdV equation \eqref{kdv0} has a solution
\begin{equation}
u(x,t)=-2\wp(x+w')-2\partial_x\int_L \varphi(x,t;l)\varrho_{l}(t)
\widetilde{ \Psi}_{x}(l)\mathrm{d}\lambda(l).\label{u:kdv}
\end{equation}
\end{theorem}

It is worth noting that the relation \eqref{kp:inte-cond},  i.e.
\begin{equation}\label{kdv:inte-cond-1}
\frac{\mathrm{d}}{\mathrm{d}x}\widetilde{ \Psi}_x(k,l)=-\widetilde{ \Psi}_{x}(k)\widetilde{ \Psi}_{x}(l)
\end{equation}
 still holds for  $\widetilde{ \Psi}_x(k)$ and $\widetilde{ \Psi}_x(k,l)$ defined in \eqref{2.26}.
 The key clue is to ensure that the right-hand side of the aforementioned equation
 converges to zero rapidly enough as $x\to +\infty$.
 Multiplying both sides of the integral equation  \eqref{integ-KdV-move}
 by the function $\widetilde{\Psi}_{s}(k)\varrho_{k}(t)$ and integrating along $L$  in the spectral parameter $k$-plane,
 the establishment of the following relation
 \begin{equation}\label{eq:kp-GLM-condit-1}
\int_x^{+\infty}\widetilde{ \Psi}_{\xi}(k)\widetilde{ \Psi}_{\xi}(l)\mathrm{d}\xi
=\widetilde{ \Psi}_x(k,l),
\end{equation}
allows the separation of parameters $k$ and $l$.
Consequently we can have a Marchenko equation
\begin{equation}\label{Marc-eq-kdv}
K(x,s,t)+F(x,s,t)+\int_x^{\infty}K(x,\xi,t)F(\xi,s,t)\mathrm{d}\xi=0,
\end{equation}
and the corresponding solution for the KdV equation:
\begin{equation}
u(x,t)=-2\wp(x+w')+2 \partial_x K(x,x,t),
\end{equation}
where
\begin{subequations}\label{KF-kdv}
\begin{align}
K(x,s,t)&=-\int_L\varphi(x,t;l)\widetilde{\Psi}_{s}(l)\varrho_{l}(t)\mathrm{d}\lambda(l),\\
F(x,s,t)&=\int_L\widetilde{\Psi}_{x}(k)\widetilde{\Psi}_{s}(k)\varrho_{k}^2(t) \mathrm{d}\lambda(k).
\end{align}
\end{subequations}

In \cite{74KM}, the Marchenko equation \eqref{Marc-eq-kdv} was directly  introduced
by assuming  $F(x,s,t)$ satisfies a kind of Lax pair of the KdV equation
( see equations (5) and (6) in  \cite{74KM}).
As discussed in section \ref{sec-2-2}, when  $k$ and $l$ take values within the interval $(0,w)$,
the validity of integral \eqref{eq:kp-GLM-condit-1} can be ensured.
At this point, the integration contour should be adjusted to encircle the discrete points $k_m\in(0,w)$.
Consequently,  $K(x,s,t)$ and $F(x,s,t)$ are written as
\begin{subequations}
\begin{align}
K(x,s,t)&=\sum_{m=1}^N \varphi(x,t;k_m)\widetilde{\Psi}_{s}(k_m)\varrho_{k_m}(t),\\
F(x,s,t)&=\sum_{m=1}^N  \widetilde{\Psi}_{x}(k_m)\widetilde{\Psi}_{s}(k_m)\varrho_{k_m}^2(t),
\end{align}
\end{subequations}
which  yield elliptic $N$-soliton solutions from the Marchenko equation \eqref{Marc-eq-kdv}.
Similar case also happened when deriving the Marchenko equation from the DL integral equation
for the usual soliton case (see equation (6) in \cite{81FA}
where they assumed $\mathrm{Im}k>0$ and $\mathrm{Im}l>0$,
which actually excludes the continuous spectrum from $F(x,t)$ for the KdV equation).

\subsubsection{Real-valued  elliptic solitons}\label{sec-3-2-3}

We assume that $L$ is a contour in the fundamental rectangular domain with vertices
$-w-w', w-w', w+w', -w+w'$,
$\{k_j\}$ belongs to the domain circled by $L$,
and the measure $\mathrm{d}\lambda(l)$ is taken as
\begin{equation}\label{measure-kdv}
\mathrm{d}\lambda(l)=\sum_{j=1}^N\frac{1}{2\pi  \mathrm{i}}\frac{ \mathrm{d}l}{l-k_j}.
\end{equation}
Then the integral equation \eqref{3.7} gives rise to a linear equation system
\begin{subequations}
\begin{equation}
(\mathbf{I+M})\mathbf{u}^{(0)}=\mathbf{r},
\end{equation}
where
\begin{align}
& \mf{M}=\left(M_{i,j}\right)_{N\times N},
~~~ M_{i,j}=\varrho_{k_i}(t)\Psi_{x}(k_i,k_j)\varrho_{k_j}(t),\\
& \mathbf{u}^{(0)}= (\varphi(x,t;k_1),  \cdots,  \varphi(x,t;k_N))^T,\\
& \mathbf{r}=(\varrho_{k_1}(t)\Psi_x(k_1),\cdots, \varrho_{k_N}(t)\Psi_x(k_N))^T.
\end{align}
\end{subequations}
Consequently, \eqref{u-int-kdv} is written as
\begin{equation}
u=-2\wp(x)-2\partial_x( \mf{r}^T\mf{u}^{(0)})= 2\partial_x^2\ln\big(\sigma(x)\tau\big),
\end{equation}
where the $\tau$-function is written as
\begin{equation}\label{tau-Hirota-KDV}
\tau=|\mf{I+M}|
=\sum_{J\subset S}\left(\prod_{i\in J}c_i\right)\Biggl(\prod_{i,j\in J\atop i<j}A_{i,j}\Biggr)
\frac{\sigma(x+2\sum_{i\in J}k_i)}{\sigma(x)\prod_{i\in J}\sigma(2k_i)}
\exp\left(\sum_{i\in J}\theta_{[e]}(x,t;k_i)\right),
\end{equation}
where $c_i=\varrho_0^2(k_i)\in \mathbb{C}$, $\sum_{J\subset S}$
means the summation runs over all subsets $J$ of $S=\{1,2,\cdots,N\}$, and
\begin{align}\label{theta-A-KDV}
 \theta_{[e]}(x,t;k_i)=-2\zeta(k_i) x -4\wp'(k_i)t, \qquad A_{i,j}
=\left(\frac{\sigma(k_i-k_j)}{\sigma(k_i+k_j)}\right)^2.
\end{align}

Similar to the KP case, starting from the integral equation \eqref{integ-KdV-move} rather than \eqref{3.7},
we will have non-singular soliton solutions
\begin{equation}\label{u:kdv-tau}
u= 2\partial_x^2\ln \big(\sigma(x+w')\tau\big),
\end{equation}
where
\begin{equation}
\tau=\tau_N
=\sum_{J\subset S}\left(\prod_{i\in J}c_i\right)\Biggl(\prod_{i,j\in J\atop i<j}A_{i,j}\Biggr)
\frac{\sigma(x+2\sum_{i\in J}k_i+w')}{\sigma(x+w')\prod_{i\in J}\sigma(2k_i)}
\exp\left(\sum_{i\in J}\theta_{[e]}(x,t;k_i)\right),
\end{equation}
and $\theta_{[e]}(x,t;k_i)$ and $A_{i,j}$ are defined as in \eqref{theta-A-KDV},
but $c_{i}$ in this turn is redefined as
\begin{equation}
c_i=\widetilde{c}_{i}\exp\left(-2\zeta(w') k_i\right),~~~ \widetilde{c}_{i}\in \mathbb{R}.
\end{equation}
The first two $\tau$-functions read
\begin{equation}\label{tau-1}
\tau_1=1+c_1\frac{\sigma(x+w'+2k_1)}{\sigma(x+w')\sigma(2k_1)}e^{\theta_{[e]}(x,t;k_1)},
\end{equation}
and
\begin{align}\label{tau-2}
\tau_2&= 1+c_1\frac{\sigma(x+w'+2k_1)}{\sigma(x+w')\sigma(2k_1)}e^{\theta_{[e]}(x,t;k_1)}
+c_2\frac{\sigma(x+w'+2k_2)}{\sigma(x+w')\sigma(2k_2)}e^{\theta_{[e]}(x,t;k_2)}\nonumber\\
&+c_1 c_2A_{1,2}\frac{\sigma(x+w'+2k_1+2k_2)}{\sigma(x+w')\sigma(2k_1)\sigma(2k_2)}
e^{\theta_{[e]}(x,t;k_1)+\theta_{[e]}(x,t;k_2)}.
\end{align}
The corresponding elliptic soliton solutions are depicted in Fig.\ref{fig:3}
where we have taken $g_2=4$, $g_3=1$ for which the corresponding half-periods are given in \eqref{w12}.

\begin{figure}[!h]
\centering
\subfigure[]{
\begin{minipage}{5.5cm}
\includegraphics[width=\textwidth]{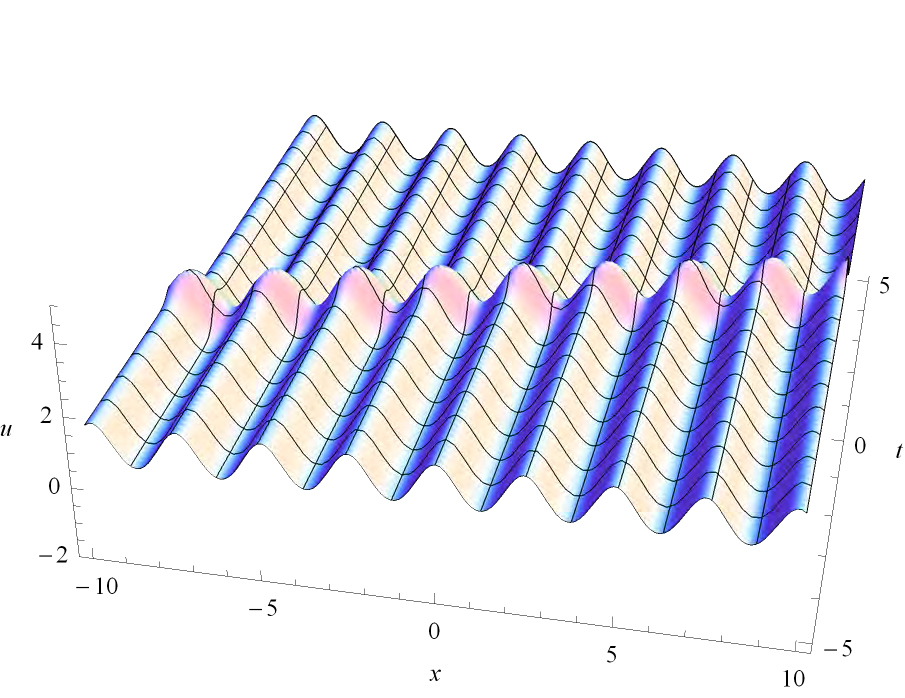}
\end{minipage}
}
\hspace{5mm}
\subfigure[]{
\begin{minipage}{5cm}
\includegraphics[width=\textwidth]{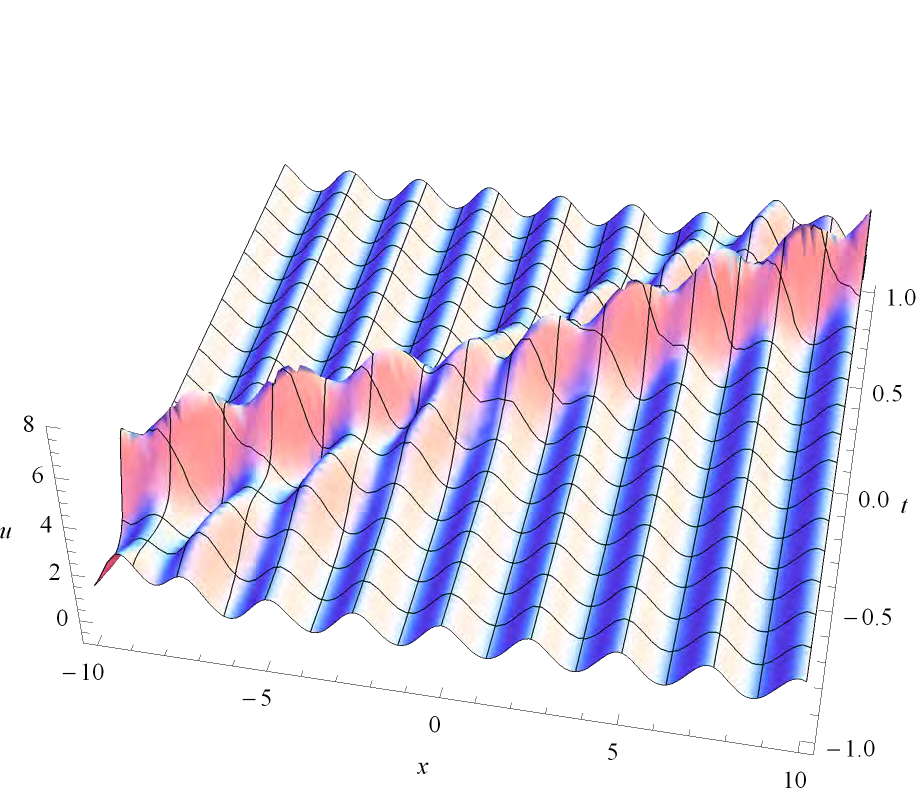}
\end{minipage}}
\caption{ Shape and motion of elliptic solitons of the KdV equation.
(a) Elliptic 1SS \eqref{u:kdv-tau} with \eqref{tau-1} for $k_1=0.9$ and $\widetilde{c}_{1}=1$.
(b) Elliptic 2SS \eqref{u:kdv-tau} with \eqref{tau-2} for $k_1=0.8$, $k_2=0.6$
and $\widetilde{c}_{1}=\widetilde{c}_{2}=1$.}
\label{fig:3}
\end{figure}

\subsection{Reduction to the Boussinesq equation}\label{sec-3-3}

\subsubsection{DL scheme}\label{sec-3-3-1}

In this section, we reduce the DL scheme and solutions of the KP equation to those of the  Boussinesq  equation,
\begin{equation}\label{BSQ}
u_{xxxx}+6u^2_x+6uu_{xx}+3u_{yy}=0,
\end{equation}
for which the elliptic cube roots $\omega_i(\delta)$ $(i=0,1,2)$ as defined
in Definition \ref{D-1} are needed, i.e. $\wp'(\delta)=\wp'(\omega_i(\delta))$,
where, especially,  we set $\omega_0(\delta)=\delta$.
The reduction can be achieved by  introducing a  special measure
in the linear integral equation of the KP equation \eqref{kp-integ} (cf.\cite{NSZ-2019}):
\begin{equation}\label{reduct-bsq}
\mathrm{d}\lambda(l,l')=\sum_{j=1}^2\delta(l'+\omega_j(l))\mathrm{d}\lambda_j(l)\mathrm{d}l'.
\end{equation}
The reduction procedure follows a similar approach to that of Theorem \ref{th:kdvDLA}.
Thus we omit it and only present the main results.

\begin{theorem}\label{th:BSQDLA}
We assume that if $\psi(x,y;k)$ solves
\begin{equation}\label{integ-BSQ}
\psi(x,y;k)+\mu_k(y)\sum_{j=1}^2\int_{L_j}
\psi(x,y;l)\gamma_{-\omega_j(l)}(y)\Psi_x{(k,-\omega_j(l))}\mathrm{d}\lambda_j(l)=\mu_k(y)\Psi_x{(k)},
\end{equation}
where $L_j$, $\mathrm{d}\lambda_j(l)$ are certain contour and measure,
\begin{align}
\mu_k(y)=\mu_0(k) e^{\wp(k)y},   \qquad
\gamma_{-\omega_j(l)}(y)=\gamma_{0}(\omega_j(l)) e^{-\wp(\omega_j(l))y},
\end{align}
and  $\omega_i(\delta)$ are the elliptic cube roots for $j=1,2$.
We also assume that differentials with respect to $x, y$ can interchange with $\int_{L_j}$,
and the homogeneous integral equation of \eqref{integ-BSQ} has only zero solution.
Then the Boussinesq  equation \eqref{BSQ} allows a solution
\begin{equation}\label{u-int-bsq}
u(x,y)=-2\wp(x)-2\partial_x\sum_{j=1}^2\int_{L_j} \psi(x,y;l)\gamma_{-\omega_j(l)}(y)
\Psi_x(-\omega_j(l))\mathrm{d}\lambda_j(l).
\end{equation}
\end{theorem}

This theorem can also be proved by verifying $\psi(x,y;k)$ and $u(x,y)$
satisfy the Lax pair of the  Boussinesq equation:
 \begin{subequations}
\begin{align}
	& P \psi(x,y;k)  = \Big(\partial_y-\partial_x^2-u(x,y)\Big)\psi(x,y;k)=0,  \\
	& M \psi(x,y;k)  =  \Big(\partial_x^3+\frac{3}{2}u(x,y)\partial_x+\frac{3}{4}u_x(x,y)
 +\frac{3}{4}\partial_x^{-1}u_y(x,y)-\frac{1}{2}\wp'(k)\Big)\psi(x,y;k)=0.
\label{BSQ-lax}	
\end{align}
\end{subequations}
Note that here \eqref{BSQ-lax} serves as a third-order spectral problem for the Boussinesq equation.

\subsubsection{Marchenko equation}\label{sec-3-3-2}

For the Boussinesq equation, to obtain real-valued elliptic soliton solutions,
we should require all three elliptic cube roots to be real, which means that $\wp'(z)=\wp'(k)$
should have three distinct real roots within the real period for a given $k\in \mathbb{R}$.
However, according to Proposition \ref{prop-A.3}, in the case where $\Delta>0$,
there are no distinct real roots for $\wp'(z)=\wp'(k)$ in $(0, 2w]$.
We exclude the case  $\Delta<0$ where distinct real roots exist,
because prefer the formulation of $w$ being real and $w'$ purely imaginary
so that we can shift the singularity to the  complex plane.
Thus, we need to  modify some settings to reduce the Marchenko equation to the Boussinesq case.

In the following, we consider an extended Boussinesq equation
which is motivated by some lattice Boussinesq equations with parameters \cite{11-Hie-BSQ,12ZZN}
(similar to the continuous case in \cite{86-Hirota-PhyD}),
where the PWFs are defined using the  roots of (see \cite{11-HZ-SIGMA,12ZZN})
\[G(z,k)=z^3-k^3+b_0(z^2-k^2)=0,\]
which allows distinct real roots for suitable parameter $b_0$
and leads to non-singular real soliton solutions for the lattice Boussinesq equation \cite{11-HZ-SIGMA}.
In the elliptic soliton case, we start from the following KP equation with a real parameter $b_0$
\begin{equation}\label{KP-b0}
u_T+ b_0u_{Y}+u_{XXX}+6uu_X+3\partial_X^{-1}u_{YY}=0,
\end{equation}
which is connected with the KP equation \eqref{eq:kp} via  coordinate transformations
\begin{equation}
  x=X, \quad  y=Y-b_0T,   \quad  t=T.
\end{equation}
The Lax pair  \eqref{kp-lax}, integral equation \eqref{kp-integ} (as well as \eqref{integ-KP-move})
and solution \eqref{u:KP} (as well as \eqref{u:KP-1}) can be accordingly and easily rewritten out in terms of $(X,Y,T)$.

With selected parameters  $b_0$ and $ l$ in $\mathbb{R}$, the equation
\begin{align}\label{cubic-roots-b0}
2\wp'(z)+b_0\wp(z)= 2\wp'(l)+b_0\wp(l)
\end{align}
allows three distinct real roots $z=\omega_j(l)$, $j=0,1,2$ in $(0, 2w]$
when $g_2, g_3\in \mathbb{R}$ and $\Delta>0$ (see Appendix \ref{App-2-2}).
In the following we always set $\omega_j(l)=l$.
Choosing the measure
\begin{align}
\mathrm{d}\lambda(l,l')=\sum_{j=1}^2\delta(l'+\omega_j(l))\mathrm{d}\lambda_j(l)\mathrm{d}l',
\end{align}
we can  reduce  the $(X,Y,T)$-counterpart of \eqref{integ-KP-move} into an  integral equation
\begin{subequations}\label{int-bsq}
\begin{align}\label{int-bsq-a}
\psi(X,Y;k)+\mu_k(Y)\sum_{j=1}^2\int_{L_j}\psi(X,Y;l)\gamma_{-\omega_j(l)}(Y)
\widetilde{\Psi}_{X}{(k,-\omega_j(l))}\mathrm{d}\lambda_j(l)
    =\mu_k(Y)\widetilde{\Psi}_{X}{(k)},
\end{align}
where
\begin{align}\label{PWF}
\mu_k(Y)=\mu_0(k)e^{\wp(k)Y},   \qquad
\gamma_{-\omega_j(l)}(Y)=\gamma_0(\omega_j(l)) e^{-\wp(\omega_j(l))Y}.
\end{align}
\end{subequations}
And from  \eqref{u:KP-1} we have
\begin{align}\label{u-bsq}
u(X,Y)=-2\wp(X+w')-2\partial_X\sum_{j=1}^2\int_{L_j} \psi(X,Y;l)
\gamma_{-\omega_j(l)}(Y)\widetilde{\Psi}_{X}(-\omega_j(l))\mathrm{d}\lambda_j(l).
\end{align}
Let us conclude the following.

\begin{theorem}\label{th:BSQ-move}
Assume $X, Y, k, b_0\in \mathbb{R}$, $\Delta>0$
and  $\omega_j(l),  j=1,2,$ are the distinct real roots of \eqref{cubic-roots-b0}.
For $\psi(X,Y;l)$ defined by \eqref{int-bsq},
the function $u(X,Y)$ defined by \eqref{u-bsq} provides a non-singular real elliptic soliton solution
for the extended Boussinesq equation
\begin{equation}\label{BSQ-b0}
(b_0u_{Y}+u_{XXX}+6uu_X)_X+3 u_{YY}=0.
\end{equation}
\end{theorem}

The term $u_{XY}$ can be  transformed into $u_{XX}$ via $\partial_Y\to\partial_Y-\frac{1}{6}b_0\partial_X$, which gives rise to the standard Boussinesq equation with a mass term
\begin{equation}
-\frac{1}{12}b_0^2u_{XX}+u_{XXXX}+6u_X^2+6uu_{XX}+3u_{YY}=0.
\end{equation}
The Theorem \ref{th:BSQ-move} can also be proven with the aid of the Lax pair
 \begin{subequations}
\begin{align}
	P \psi(X,Y;k)  = & \Big(\partial_Y-\partial_X^2-u(X,Y)\Big)\psi(X,Y;k)=0,  \\
	M \psi(X,Y;k)  = &
 \Big(\partial_X^3+\frac{3}{2}u(X,Y)\partial_X+\frac{3}{4}u_X(X,Y)+\frac{3}{4}\partial_X^{-1}u_Y(X,Y) \nonumber\\
 &-\frac{1}{2}\wp'(k)-\frac{1}{4}b_0\wp(k)+\frac{1}{4}b_0\partial_Y\Big)\psi(X,Y;k)=0.
\label{ex-BSQ-lax}	
\end{align}
\end{subequations}
To have a Marchenko equation for the extended  Boussinesq equation \eqref{BSQ-b0},  we introduce
\begin{subequations}\label{KF-bsq-1}
\begin{align}
K_i(X,s,Y)&= -\int_{L_i}\psi(X,Y;l)\widetilde{\Psi}_{s}(-\omega_i(l))
\gamma_{-\omega_i(l)}(Y)\mathrm{d}\lambda_i(l),\\
F_i(X,s,Y)&= \int_{L_i}\widetilde{\Psi}_{X}(k)\mu_k(Y)\widetilde{\Psi}_{s}(-\omega_i(k))
\gamma_{-\omega_i(k)}(Y) \mathrm{d}\lambda_i(k).
\end{align}
\end{subequations}
Multiplying both sides of the integral equation \eqref{int-bsq} by
$\widetilde{\Psi}_{s}(-\omega_i(k))\gamma_{-\omega_i(k)}(Y)$ and integrating along the path $L_i$,
similar to   the KdV case (see Sec.\ref{sec-3-2-2}),
with  the help of the  valid integral
\begin{equation}
\int_x^{+\infty}\widetilde{ \Psi}_{\xi}(k)\widetilde{ \Psi}_{\xi}(-\omega_i(l))\mathrm{d}\xi
=\widetilde{ \Psi}_x(k,-\omega_i(l)),
\end{equation}
we can get
\begin{equation}
K_i(X,s,Y)+F_i(X,s,Y)+\int_x^{\infty}F_i(\xi,s,Y)\left(K_1(X,\xi,Y)+K_2(X,\xi,Y)\right)\mathrm{d}\xi=0
\end{equation}
for $i=1,2$.
Next, setting
$$K(X,s,Y)=\sum_{i=1}^2K_i(X,s,Y),~~~  F(X,s,Y)=\sum_{i=1}^2F_i(X,s,Y),$$
we arrive at a Marchenko equation, i.e.
\begin{equation}
K(X,s,Y)+F(X,s,Y)+\int_x^{\infty}K(X,\xi,Y)F(\xi,s,Y)\mathrm{d}\xi=0,
\end{equation}
for the extended Boussinesq equation \eqref{BSQ-b0},
and its solution can be given by
\begin{equation}
u(X,Y)=-2\wp(X+w')+2 \partial_X K(X,X,Y).
\end{equation}

In some special cases, one can get real solutions from the above Marchenko equation.
However, we prefer to derive solutions directly from the integral equation \eqref{int-bsq}.
We now proceed to the next step for real elliptic soliton solutions.

\subsubsection{Real-valued  elliptic solitons}\label{sec-3-3-3}

To derive  elliptic soliton solutions of the extended Boussinesq equation \eqref{BSQ-b0},  we choose a measure
in \eqref{int-bsq-a} as
\begin{equation}\label{measure-bsq}
\mathrm{d}\lambda_j(l)=\frac{1}{2\pi \mathrm{i}}\sum_{\beta=1}^{N}\frac{\mathrm{d}l}{l-k_\beta},~~j=1,2,
\end{equation}
and  $L_j$ being the contour such that $\{k_i\}$ belong to the domain circled by $L_j$.
Then the integral equation \eqref{int-bsq-a} gives rise to a linear equation system
\begin{equation}
(\mathbf{I+M})\mathbf{u}^{(0)}=\mathbf{r},
\end{equation}
where
\begin{equation*}
\begin{split}
&\mf{M}=\left(M_{\alpha,\beta}\right)_{N\times N},    \quad
M_{\alpha,\beta}=\mu_{k_{\alpha}}(Y)\sum_{j=1}^2\widetilde{\Psi}_{X}(k_{\alpha},-\omega_j(k_{\beta}))
\gamma_{-\omega_j(k_{\beta})}(Y),
~~\alpha,\beta=1,\cdots, N,\\
& \mathbf{u}^{(0)}=(\psi(X,Y;k_{1}),~\psi(X,Y;k_{2}),~\cdots,~\psi(X,Y;k_{N}))^T,\\
&\mathbf{r}=(\mu_{k_{1}}(Y)\widetilde{\Psi}_X(k_{1}),~\mu_{k_{2}}(Y)\widetilde{\Psi}_X(k_{2}),~\cdots,  ~
\mu_{k_{N}}(Y)\widetilde{\Psi}_X(k_{N}))^T,
\end{split}
\end{equation*}
and the solution \eqref{u-bsq} is presented as
\begin{equation}
u(X,Y)=-2\wp(X+w')-2(\mf{s}^T (\mathbf{I+M})^{-1}\mathbf{r})_X
=2\partial_X^2\ln \big(\sigma(X+w')\tau\big),
\end{equation}
where $\mathbf{r}$ is defined as above,
\begin{equation}
\mathbf{s}=\left(\sum_{j=1}^2\gamma_{-\omega_j(k_{1})}(Y)\widetilde{\Psi}_X(-\omega_j(k_{1})),
~\cdots, ~
\sum_{j=1}^2\gamma_{-\omega_j(k_{N})}(Y)\widetilde{\Psi}_X(-\omega_j(k_{N}))\right)^T,
\end{equation}
and the $\tau$ function is written as
\begin{equation}\label{tau-bsq}
\tau =|\mf{I+M}|.
\end{equation}
The explicit form of this $\tau$ can be expressed as
\begin{equation}
\tau=\sum_{J\subset S}\sum_{\nu_{J\alpha}}
\left[ \Bigg(\prod_{i<j \in J}A_{i,j}\Biggr) \times F_J \times
\Bigg(\prod_{i\in J} \sum_{\alpha=1}^2 \nu_{i\alpha}c_{i\alpha} \mathrm{exp}\big(\theta_{[e]}(X,Y;k_i,\omega_\alpha(k_i))\big)\Biggr)\right ],
\end{equation}
where $\nu_{i\alpha} \in \{0,1\}$, $c_{i\alpha}=\mu_0(k_i)\gamma_0(\omega_\alpha(k_i))$
for $i=1,2,\cdots,N$ and $\alpha=1,2$,
\begin{align*}
& \theta_{[e]}(X,Y;k_{i},\omega_\alpha(k_i))=
(-\zeta(k_{i})+\zeta(\omega_\alpha(k_{i})))X+(\wp(k_{i})-\wp(\omega_\alpha(k_{i})))Y
-\zeta(w')(k_{i}-\omega_\alpha(k_{i})),\\
& F_J=\frac{\sigma(X+w'+\sum_{i\in J}(k_i-\sum_{\alpha=1}^2\nu_{i\alpha}\omega_\alpha(k_i)))}
{\sigma(X+w')\prod_{i\in J}\sigma(k_i-\sum_{\alpha=1}^2\nu_{i\alpha}\omega_\alpha(k_i))},\\
&A_{i,j}=\frac{\sigma(k_i-k_j)\sigma(-\sum_{\alpha=1}^2\nu_{i\alpha}\omega_\alpha(k_i)+\sum_{\alpha=1}^2
\nu_{j\alpha}\omega_\alpha(k_j))}
{\sigma(k_i-\sum_{\alpha=1}^2\nu_{j\alpha}\omega_\alpha(k_j))\sigma(-\sum_{\alpha=1}^2
\nu_{i\alpha}\omega_\alpha(k_i)+k_j)},
\end{align*}
$\sum_{J\subset S}$ means the summation runs over all subsets $J$ of $S=\{1,2,\cdots,N\}$,
and the summation over $\nu_{J\alpha}$ is performed as a summation for all the cases
where $\nu_{i1}, \nu_{i2} \in \{0,1\}$ for $i\in J$ but subject to  $\nu_{i1}+\nu_{i2}=1$.

As an example, we take $g_2=4$, $g_3=1$, and $b_0=13$.
The nonsingular real-valued elliptic 1SS  for equation \eqref{BSQ-b0}  is provided by
\begin{equation}
 u_1(X,Y) =2\partial_X^2
\ln\Big(\sigma(X+w')+ \sum_{\alpha=1}^2  c_{1\alpha}\frac{\sigma(X+k_{1}-\omega_\alpha(k_{1})+w')}
  {\sigma(k_{1}-\omega_\alpha(k_{1}))}e^{\theta_{[e]}(X,Y;k_{1},\omega_\alpha(k_1))}\Big) .
  \label{1ss-bsq-1}
\end{equation}
This solution is depicted in Fig.\ref{fig:4}.
\begin{figure}[!h]
\centering
\begin{minipage}{5.5cm}
\includegraphics[width=\textwidth]{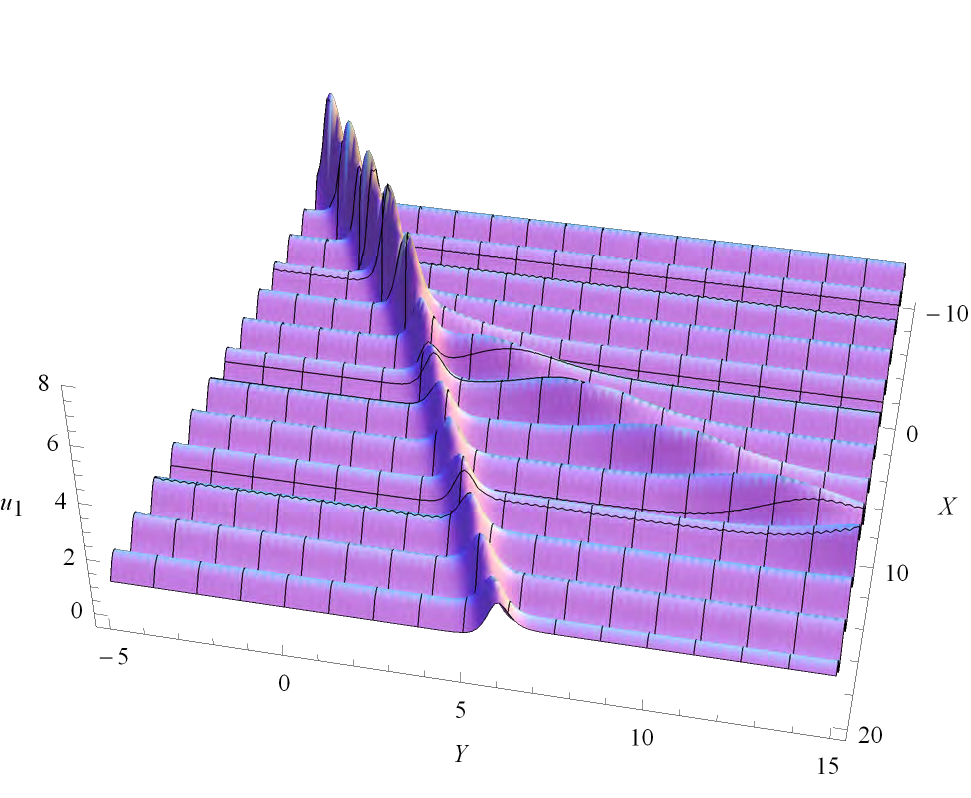}
\end{minipage}
\caption{Shape and motion of  elliptic 1SS   \eqref{1ss-bsq-1} for $c_{11}=c_{12}=1$, $k_1=0.7$,
$\omega_1(k_1)=0.37233792...$, $\omega_2(k_1)=1.37905146...$.}
\label{fig:4}
\end{figure}

Note that the solutions can also  involve only one elliptic cubic root,
for example, setting $\gamma_0(\omega_2(k_i))=0$.
In such a case, the corresponding elliptic 1SS and 2SS for the extended Boussinesq equation \eqref{BSQ-b0}
are given by
\begin{equation}\label{BSQ-sol}
u(X,Y) = 2 \partial_X^2 \ln\big(\sigma(X+w') \tau\big)
\end{equation}
where, respectively,
\begin{align}
\tau=1+ c_{11}  \frac{\sigma(X+k_{1}-\omega_1(k_{1})+w')}
  {\sigma(X+w')\sigma(k_{1}-\omega_1(k_{1}))}e^{\theta_{[e]}(X,Y;k_{1},\omega_1(k_1))}
\label{1ss-bsq}
\end{align}
and
\begin{align}
 &\tau= 1+c_{11}\frac{\sigma(X+k_{1}-\omega_1(k_{1})+w')}
{\sigma(X+w')\sigma(k_{11}-\omega_1(k_{1}))}e^{\theta_{[e]}(X,Y;k_{1},\omega_1(k_1))}\nonumber\\
&\quad +c_{21}\frac{\sigma(X+k_{2}-\omega_1(k_{2})+w')}
  {\sigma(X+w')\sigma(k_{2}-\omega_1(k_{2}))}e^{\theta_{[e]}(X,Y;k_{2},\omega_1(k_2))}
    +c_{11}c_{21}\frac{\sigma(k_{1}-k_{2})\sigma(\omega_1(k_{1})-\omega_1(k_{2}))}
  {\sigma(k_{1}-\omega_1(k_{2}))\sigma(\omega_1(k_{1})-k_{2})}\nonumber \\
&\quad \times
  \frac{\sigma(X+k_{1}-\omega_1(k_{1})+k_{2}-\omega_1(k_{2})+w')}
  {\sigma(X+w')\sigma(k_{1}-\omega_1(k_{1}))\sigma(k_{2}-\omega_1(k_{2}))}
  e^{\theta_{[e]}(X,Y;k_{1},\omega_1(k_1))+\theta_{[e]}(X,Y;k_{2},\omega_1(k_2))}.
\label{2ss-bsq}
\end{align}
The illustrations are presented   in Fig.\ref{fig:5}.

\begin{figure}[!h]
\centering
\subfigure[]{
\begin{minipage}{5.5cm}
\includegraphics[width=\textwidth]{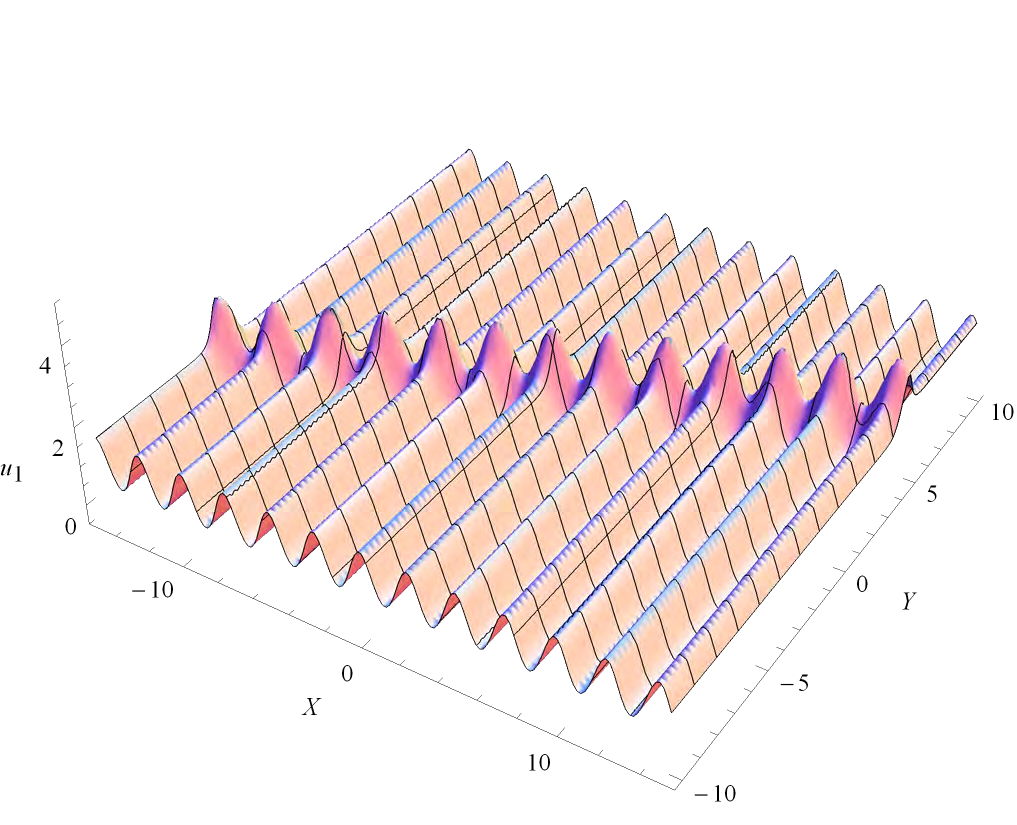}
\end{minipage}
}
\hspace{5mm}
\subfigure[]{
\begin{minipage}{5.4cm}
\includegraphics[width=\textwidth]{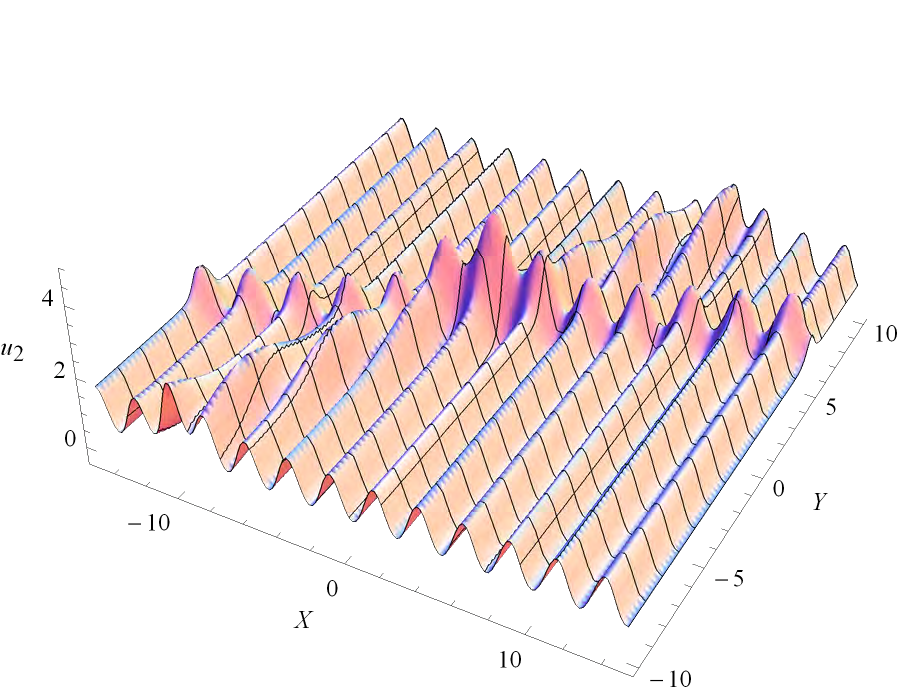}
\end{minipage}}
\caption{ Shape and motion of elliptic soliton solutions of the extended Boussinesq equation \eqref{BSQ-b0}
with $b_0=13$.
(a) 1SS \eqref{BSQ-sol} with \eqref{1ss-bsq} for  $c_{11}=1$, $k_{1}=0.7$, $\omega_1(k_{1})=0.3723379...$.
(b) 2SS \eqref{BSQ-sol} with \eqref{2ss-bsq}  $c_{11}=c_{21}=1$, $k_{1}=0.65$,
$\omega_1(k_{1})=0.3817591...$, $k_{2}=1.4$, $\omega_2(k_{2})=0.6746367...$.}
\label{fig:5}
\end{figure}

\section{Conclusions}\label{sec-4}

In this paper, we have established the elliptic DL schemes and Marchenko equations for the KP,
KdV and Boussinesq equations.
Non-singular real elliptic soliton solutions for these equations were obtained.
The DL scheme of the KP equation consists of an integral equation (e.g.\eqref{kp-integ})
and a formula of solution (e.g.\eqref{u:KP}).
This is confirmed by checking the Lax pair.
Since the Marchenko equation involves integration on the real axis,
we have discussed real-valuedness of the Weierstrass functions in Appendix \ref{App-1}.
Based on these discussions, we can choose parameters $k$ and $l'$
such that the integration \eqref{eq:kp-GLM-condit} on $(x,\infty)$ makes sense.
Then we were able to construct a  Marchenko equation for elliptic soliton solutions (see \eqref{eq:GLM-KP}).
To have non-singular real solutions, in addition to utilizing real-valuedness of the Weierstrass functions,
we also employed a trick, $x\to x+w'$, to shift zeros of $\sigma(z)$ from the real axis to the half complex plane.
As a result, we obtained $\tau$ function for non-singular real elliptic solitons of the KP equation.
The illustration shows line solitons on a periodic background.
We also obtained resonant Y-shape elliptic line solitons.

The concept of elliptic $N$th  roots of unity can be used to reduce the elliptic DL scheme and $\tau$ function
of the KP equation to those of the KdV and Boussinesq equations.
However, in this paper, we have only obtained formal Marchenko equations for the KdV and Boussinesq equations.
This is because, for the KdV equation,
we must consider simultaneously the convergence of the integral \eqref{eq:kp-GLM-condit-1}
and the significance of the contour integrals in \eqref{KF-kdv}.
Same problems are encountered for the Boussinesq equation.
Note that if we are only concerned with elliptic solitons,
which correspond to discrete eigenvalues, the Marchenko equations can always be obtained,
and they are in the same form as for the usual solitons.

With respect to constructing elliptic soliton solutions, the DL approach allows more flexibility,
compared with using the Marchenko equation (cf.\cite{74KM} for the KdV).
One may also try other measures to get solutions of other types,
e.g. multiple-pole type elliptic solitons, which have not been investigated before.
In addition, the DL scheme is closely related to the Cauchy matrix structure of soliton solutions.
For usual soliton solutions, such structures are formulated by a Sylvester equation
together with dispersion relations (corresponding to the equations under investigation,
see e.g. \cite{13ZZ,14-XuZhangZhao}).
For elliptic soliton solutions, it seems there is no Sylvester type equations available  (see also \cite{IMRN2010}).
Note that the Cauchy matrix approach has been well developed for the usual soliton case,
e.g. \cite{09NAH,13ZZ,14-XuZhangZhao,17-SunZhangNijhoff,23-LQZ}.
In particular, it allows explicit multiple pole solutions by solving the Sylvester equation.
How to get multiple-pole elliptic solitons in terms of elliptic Cauchy matrix is also an interesting problem
for consideration.

\vskip 20pt
\subsection*{Acknowledgments}
This project is supported by the NSF of China (Nos.12271334, 12411540016, 12301309, 12471235).

\appendix

\section{Real-valuedness of the Weierstrass functions}\label{App-1}

For this part we mainly refer to  \cite{book-Pastras,book-AbraSte}.

The $\wp(z)$  and its derivative compose points on the elliptic curve \eqref{ell-cur}, i.e.
\begin{equation}\label{ell-cur}
(\wp'(z))^2=4\wp^3(z)-g_2\wp(z)-g_3=4(\wp(z)-e_1)(\wp(z)-e_2)(\wp(z)-e_3),
\end{equation}
where  $e_i=\wp(w_i), i=1,2,3, $ and $w_1=w$, $w_2=w+w'$, $w_3=w'$.
Non-zero discriminant $\Delta=g_2^3-27g_3^2$ ensures the smoothness of the curve.
The other two Weierstrass functions, $\zeta(z)$ and $\sigma(z)$, are  quasi-periodic, see \eqref{periodicity}.

In the following we consider the case where the moduli $g_2$ and $g_3$ are real.
First, we have the following.

\begin{Proposition}\label{prop-A.1}
    When the moduli $g_2$, $g_3$ are real, $\sigma(z)$,  $\zeta(z)$ and $\wp(z)$ take real values on the real axis.
\end{Proposition}

\begin{proof}
In fact, when $g_2,g_3\in \mathbb{R}$,  the three Weierstrass functions satisfy the following property
(see page 631 of \cite{book-AbraSte}):
\begin{equation}
\overline{f(z)}=f(\overline{z}),~~ z\in \mathbb{C},
\end{equation}
where `bar' denotes complex conjugate, and $f$ can be either $\wp$,  $\zeta$ or $\sigma$.
When $z\in \mathbb{R}$, we immediately prove the proposition.

\end{proof}

For the real moduli $g_2$ and $g_3$, the non-degenerate discriminant $\Delta$
can be either  positive or negative, which determines different properties of the
Weierstrass functions. We collect some of them in the following Proposition.

\begin{Proposition}\label{prop-A.2}
We list some elements and their values when the moduli $g_2$ and $g_3$ are real.
They rely on $\Delta$.

\begin{table}[H]
	\centering
\begin{tabular}{p{2.7cm}|p{5.8cm}|p{5.5cm}}
\hline
~         &positive discriminant $\Delta>0  $         & negative discriminant $\Delta<0$   \\	
\hline	
 (1) Fundamental
 half periods       &$w_1$ is real, $w_3$ is  purely imaginary.
                &  $w_1=a+\mathrm{i}b$  and $w_3=a-\mathrm{i}b$ ($a>0$, $b>0$). \\
\hline
(2) Values at half-periods
        &   $e_1>0\geq e_2>e_3$;
       $\zeta(w_1)>0$,  $\mathrm{i}\zeta(w_3) \in \mathbb{R}$;
        $e_1^2>\frac{g_2}{12}$, $e_3^2>\frac{g_2}{12}$.
           \newline
          ($g_3\geq 0$,  equality when $g_3=0$).
                                 & $e_2\geq 0$, $e_1 =-\alpha+\mathrm{i}\beta$,
                                   $e_3=\overline{e_1}$, ($\alpha\geq 0$, $\beta>0$);
                                   $\zeta(w_2)>0$,    $\mathrm{i}\zeta(w_3-w_1)  \in \mathbb{R}$;
                                    $e_2^2 >\frac{g_2}{12}$;
                                       \newline
                                     ($g_3\geq 0$, equality when $g_3=0$). \\
\hline
(3) Real-value of $ \wp(z)$
(modulo the period lattice)
&  $z\in (0,2w_1)$, $\wp(z)\geq e_1$;  \newline
 $z\in (0,2w_3)$, $\wp(z)\leq e_3$; \newline
     $z\in [w_1,w_1+2w_3]$, $e_2\leq\wp(z)\leq e_1$;
       $z\in [w_3,w_3+2w_1]$, $e_3\leq\wp(z)\leq e_2$.
                           & $\wp(z)$ takes real values only on the real and                                                                                                                                      imaginary axes. \newline
                               $z\in (0,2w_2)$, $\wp(z)\geq e_2$;  \newline
                                 $z\in (0,2w_1-2w_3)$, $\wp(z)\leq e_2$. \\
\hline
(4) Real-value of $\Psi_z(a)$
                     &$z\in \mathbb{R}$ and $a\in (0,2w_1)$, $(w_3,w_3+2w_1)$, (modulo the period lattice).
                                           & $z\in \mathbb{R}$ and $a\in (0,2w_2)$, (modulo the period lattice). \\
\hline
\end{tabular} 		
\end{table}
We also note that in the above table $(c,d)$ (or $[c,d]$) for $c,d\in \mathbb{C}$ means the segment
on the complex plane from point $c$ to $d$, rather than an interval on the real axis.
\end{Proposition}

\begin{proof}
(1) See  page 630 in \cite{book-AbraSte}.

(2) See page 633 in \cite{book-AbraSte}.
In the table we only list the case $g_3\geq 0$. The case  $g_3<0$  can be deduced by using the relations
  \begin{equation}
  -\wp(z;g_2,g_3)=\wp(\mathrm{i}z;g_2,-g_3), \qquad   -\mathrm{i}\zeta(z;g_2,g_3)=\zeta(\mathrm{i}z;g_2,-g_3).
  \end{equation}
When $g_3<0$, we  denote the half-periods of $\wp(z; g_2,g_3)$ by $\widetilde{w}_j \,(j=1,3)$,
and  $\tilde{e}_j=\wp(\tilde{w}_j;g_2,g_3)\,(j=1,2,3)$,
where $\tilde{w}_2=\tilde{w}_1+\tilde{w}_3$.
When $\Delta>0$,  one can take
$\tilde{w}_1=-\mathrm{i}w_3>0$, $\tilde{w}_3=\mathrm{i}w_1$ (see page 105 of \cite{book-Akhiezer}),
thus we have
\begin{align}
&\tilde{e}_1>\tilde{e}_2>0>\tilde{e}_3, \quad \tilde{e}^2_1>\frac{g_2}{12},
\quad \tilde{e}^2_3>\frac{g_2}{12},\quad
\zeta(\tilde{w}_1;g_2,g_3)\in \mathbb{R},  \quad  \mathrm{i} \zeta(\tilde{w}_3;g_2,g_3)>0.
\end{align}
When $\Delta<0$,  one may take $\tilde{w}_1=\mathrm{i}w_3$, $\tilde{w}_3=-\mathrm{i}w_1$,  thus we have
\begin{align}\label{A-5}
&
\tilde{e}_2<0, \quad \tilde{e}^2_2>\frac{g_2}{12}, \quad \zeta(\tilde{w}_2;g_2,g_3)\in \mathbb{R}, \quad
 \mathrm{i} \zeta(\tilde{w}_1-\tilde{w}_3;g_2,g_3)>0.
\end{align}

For (3) and (4), see pages 34 and 59 of  \cite{book-Pastras}, respectively.

\end{proof}

\section{Discussions on the real elliptic cube roots of unity}\label{App-2}

\subsection{Distinct real cube roots}\label{App-2-1}

Let us fisrt discuss the roots of
\begin{equation}\label{pzk}
 \wp'(z)=\wp'(k).
\end{equation}

\begin{Proposition}\label{prop-A.3}
Assume $g_2, g_3, k\in \mathbb{R}$.
We have the following two results for the real roots of \eqref{pzk}.
\\
(1) When $\Delta>0$,
equation \eqref{pzk} does not have distinct real roots within the real period.\\
(2) When $\Delta<0,   g_2>0,  g_3<0$,  equation \eqref{pzk} can have three distinct real roots
with suitable values  of the parameter $k$ in a fundamental domain.
\end{Proposition}

\begin{proof}
(1) In fact, in case equation \eqref{pzk} having  distinct real roots within the real period $(0,2w_1]$,
$\wp''(z)$ will have at least  one real zero.
However, recalling $\wp''(z)=6\wp^2(z)-\frac{g_2}{2}$ and $\wp^2(z)\geq e_1^2>\frac{g_2}{12}$
when $\Delta>0$ and $z\in(0, 2w_1)$ (see Proposition \ref{prop-A.2}),
we know that $\wp''(z)>0$. This means there is no zeros for $\wp''(z)$ in $(0, 2w_1)$.\\
(2) In the case $\Delta<0$,   to ensure that $\wp''(z)=0$ has real roots, we only consider the case where $g_2>0$.
When $ g_3\geq0$,
according to the  Proposition \ref{prop-A.2},
  $\wp^2(z)\geq e_2^2>\frac{g_2}{12}$ always holds true within the real period.
 For the case $g_3<0$, from the equation \eqref{A-5}, it is known that $\tilde{e}_2<-\sqrt{\frac{g_2}{12}}$.
Due to $\lim_{z\to0^+}\wp(z)=\lim_{z\to0^-}\wp(z)=+\infty$,
it can be deduced that $\wp(z)=\pm\sqrt{\frac{g_2}{12}}$
has four distinct real roots within the real period  $(-\tilde{w}_2,\tilde{w}_2)$.
 In other words,  $\wp'(z)$ has four  extremum points, which can be denoted as
$-z_1, -z_2, z_2, z_1$ in ascending order within the real period.
Since $\lim_{z\to0^-}\wp'(z)=+\infty$, $\lim_{z\to0^+}\wp'(z)=-\infty$
and $\wp'(\tilde{w}_2)=\wp'(-\tilde{w}_2)=0$,
it can be deduced that by selecting the appropriate $k$ such that  $-\wp'(z_2)<|\wp'(k)|<-\wp'(z_1)$ holds,
and then equation \eqref{pzk} can have three distinct real roots.

\end{proof}

\subsection{Distinct real roots of equation \eqref{cubic-roots-b0}}\label{App-2-2}

We discuss the values of parameters  $b_0$ and $l$ such that the equation \eqref{cubic-roots-b0} has three different real roots.
Let us define $g(x)=2\wp'(x)+b_0\wp(x)$, then  the equation \eqref{cubic-roots-b0} can be expressed as $g(z)=g(l)$.
 Without loss of generality, we consider the case for $\Delta>0$, $g_3\geq0$ and $b_0>0$.
 Similar to the proof of Proposition \ref{prop-A.3},
the selection of  $b_0$ should ensure that $g(x)$ has two extremum points,
which also implies that $g'(x)$ has  two zeros.
When $x\in[w_1,2w_1)$, due to $\wp'(x)\geq0$ and Proposition \ref{prop-A.2},
we have $g'(x)=12\wp^2(x)-g_2+b_0\wp'(x)>0$.
Hence, the  zeros of $g'(x)$ belong to $(0,w_1)$.
Next, we determine the value of $b_0$ using the root $x_0$ of $g''(x)=0$, i.e.,
\begin{equation*}
b_0=-\frac{2\wp'''(x_0)}{\wp''(x_0)},
\end{equation*}
while also ensuring that $g'(x_0)<0$.
Through the series expansion of the aforementioned equation near the origin,
 it can be observed that as $x_0\to 0^{+}$, the value of $b_0$ gradually increases.
Taking $x_0=\frac{w_1}{2}$ as an example,
and with the aid of $\wp(\frac{w_1}{2})=e_1+H_1$, $\wp'(\frac{w_1}{2})=-2H_1\sqrt{2H_1+3e_1}$,
where $H_1=\sqrt{3e_1^2-\frac{g_2}{4}}$, we can obtain the value of $b_0$
\begin{equation*}
b_0
=12\frac{e_1+H_1}{\sqrt{2H_1+3e_1}},
\end{equation*}
and then we have
\begin{equation*}
  g'(x_0)=2\left(\wp''(x_0)-\frac{\wp'''(x_0)\wp'(x_0)}{\wp''(x_0)}\right)=-4H_1(4H_1+3e_1)<0.
\end{equation*}
Meanwhile, as $x_0\to 0^{+}$, $g'(x_0)<0$.
Therefore, when we  take  $-2\wp'''(\frac{w_1}{2})/\wp''(\frac{w_1}{2})\leq b_0\ll\infty$,
we can arrive at the  numerical values of the  two extremum points  of $g(x)$.
By ensuring that $g(l)$ is situated  between the local maximum and minimum values,
we can obtain three zeros  $\omega_0(l)=l$, $\omega_1(l)$ and $\omega_2(l)$ of   equation \eqref{cubic-roots-b0}.
The numerical results can be calculated using \emph{Mathematica}.

\section{On elliptic solitons}\label{App-3}

The authors are grateful to the referee for the comments on elliptic solutions and related references,
which inspire us to complete this brief review on elliptic solitons.

The investigation of ``elliptic solitons'' dates back to the research of Krichever \cite{K-FAA-1980}
on  elliptic solutions of the KP equation and the KdV equation.
He proposed a  criterion for such solutions that
the solutions should be elliptic in the spatial variable $x$.
For the KP equation, Krichever proved its elliptic solutions
should take the form (\cite{K-FAA-1980}, Theorem 4)
\begin{equation}\label{u-KP-ell}
u=\mathrm{const.}+\sum_{j=1}^n \wp(x-x_j(y,t)),
\end{equation}
where $\{x_j(y,t)\}$ obey some dynamical system of finite dimension.
He also showed that \eqref{u-KP-ell} can be expressed in terms of $\theta$-function (\cite{K-FAA-1980}, Corollary 3).
The reduction to the KdV equation recovers the 2-gap elliptic solution in \cite{DN-JETP-1974}
and the results of the  paper \cite{AMM-CPAM-1977}
that for the first time links the elliptic solutions of integrable equations
with the theory of finite-dimensional dynamic systems.

The term ``elliptic soliton'' is introduced by Verdier in \cite{V-AA-book-1988},
where he described solitons using the language of algebraic geometry.
In his definition, ``algebraic solitons'' are governed by some linear algebraic subgroup $G$,
while when $G$ is an elliptic curve, the solutions are called ``elliptic solitons''.
Verdier pointed out in \cite{V-AA-book-1988} that the elliptic solitons of the KP equation
must be of the form
\begin{equation}\label{u-KP-ell-Verdier}
u=\mathrm{const.}+2 \sum_{j=1}^n m_j\wp(x-x_j(y,t)),
\end{equation}
where $m_j$ are triangular numbers,
and some new elliptic source potentials (without time evolution) for the KdV equation were listed out in
 \cite{V-AA-book-1988}.
Verdier's research 
was elaborated later with his collaborator Treibich  in
\cite{T-DMJ-1989,TV-G-book-1990,TV-CRASP-1990,T-AAM-1994}
by using the tool of ``tangential cover''.
In particular, in \cite{TV-CRASP-1990} they presented the following elliptic source potential for the KdV equation:
\begin{equation}\label{u-TV}
u= \sum_{j=0}^3 a_j(a_j+1)\wp(x-x_j), ~~~~~ \sum_{j=0}^3a_j(a_j+1)=2n,
\end{equation}
where $n$ is the degree of the tangential cover.
This is more general than the Lam\'e-Ince potential \cite{I-PRSE-1940}
which reads $u= g(g+1)\wp(x), ~ g\in \mathbb{N}$.
The function \eqref{u-TV} is known as the Treibich-Verdier potential
and now the Darboux-Treibich-Verdier (DTV) potential.
It is Matveev and Smirnov \cite{MS-LMP-2006} who pointed out that
Darboux found the same potential in 1882 \cite{D-CRASP-1882}
but given in terms of Jacobi elliptic functions.
They also showed in \cite{MS-LMP-2006} that the Darboux equation is a canonical form the de Sparre equation
\cite{dS-AM-1883ab}.
There are other mathematicians and physicists joined the
research of the DTV potential, e.g. 
\cite{BE-FAA-1989,GW-MZ-1995,S-MN-1989,S-AAM-1994,S-CRMP-2002,T-CMP-2003,V-LMP-2011}.
Later, finite-gap elliptical  potentials of a more general form (e.g. \cite{S-LMP-2006}, Eq.(0.10))
of the  Schr\"odinger spectral problem
$
\psi_{xx}(x,\lambda)+u(x)(x,\lambda)=\lambda \psi(x,\lambda)
$
were found \cite{S-LMP-2006,T-RMS-2001}.
It is also remarkable that Gesztesy and Weikard \cite{GW-AM-1996}
investigated the Schr\"odinger spectral problem
by utilizing a theorem of Picard which
concerned with the existence of solutions
that are elliptic of the  $n$th-order ordinary differential equations
with elliptic coefficients.
They proved that all elliptic KdV algebro-geometry potentials
are equivalent to the Picard potentials \cite{GW-AM-1996}.
As examples they gave several elliptic potentials corresponding to genus $g=4,5$ 
(\cite{GW-AM-1996}, Eqs.(5.40) and (5.43); also \cite{GW-MZ-1995}, Eqs.(3.86)-(3.89)).
It is also worth mentioning that Krichever et.al. investigated the case of matrix KP equation \cite{KBBT-AMST-book-1995}
and the analogue research of \cite{T-DMJ-1989} for matrix case was developed by Treibich in \cite{T-DMJ-1997}.

In addition, Belokolos and Enol'skii  considered  reduction from the general finite-gap solution to the DTV potential
\cite{BE-FAA-1989}. Later, they \cite{BE-AAM-1994} described a technique on how  elliptic finite-gap potentials are
derived from the Matveev-Its formula for the finite-gap potential.
Besides, Enol'skii and Kostov \cite{EK-AAM-1994} discussed the  conditions on the moduli of the
corresponding algebraic curve under which the general finite-gap solution reduces to elliptic solitons.
They indicated a completely algorithmic construction of the elliptic soliton moduli.
There are also studies on the elliptic potentials related to the Ablowitz-Kaup-Newell-Segur  spectral problem,
for example, \cite{S-RASSM-1995} from Smirnov and \cite{GW-AM-1998} by Gesztesy and Weikard.

The solutions we derived in this paper can be considered as special excitations of elliptic potentials 
without increasing genus.
In illustration they are  solitons living on an elliptic background.
Strictly speaking, these solutions are not elliptic in $x$ any longer.
We follow \cite{IMRN2010} and still employ the name ``elliptic solitons'' in this paper.


\vskip 20pt


\begin{thebibliography}{99}


\bibitem{82AF} M.J. Ablowitz, A.S. Fokas,
        A direct linearization associated with the Benjamin-Ono equation,
        in {Mathematical Methods in Hydrodynamics and Integrability in Dynamical Systems},
        Eds: M. Tabor, Y.M. Treve,
        AIP Conference Proceedings, 88 (1982) 229-236.

\bibitem{83-AFA-BO}  M.J. Ablowitz,  A.S. Fokas, R.L. Anderson,
        The direct linearizing transform and the Benjamin-Ono equation,
        Phys. Lett. A, 93 (1983) 375-378.
        
\bibitem{ABS03} V.E. Adler, A.I. Bobenko, Y.B. Suris,
         Classification of integrable equations on quad-graphs, the consistency approach,
         Commun. Math. Phys., 233 (2003) 513-543.

\bibitem{AMM-CPAM-1977} H. Airault, H.P. McKean, J. Moser,
         Rational and elliptic solutions of the Korteweg-de Vries equation and a related many-body problem,
         Commun. Pure Appl. Math., 30 (1977) 95-148.

\bibitem{book-Akhiezer} N.I. Akhiezer,
             Elements of the Theory of Elliptic Functions.
             Translated from the Russian edition by H.H. McFaden,
             Translations of Mathematical Monographs 79,
             Amer. Math. Soc., Providence, 1990.

\bibitem{BE-FAA-1989} E.D. Belokolos, V.Z. Enol'skii,
           Verdier's elliptic solitons and the Weierstrass reduction theory,
           Funct. Aanl. Appl., 23 (1989) 46-47.

\bibitem{BE-AAM-1994} E.D. Belokolos, V.Z. Enol'skii,
         Reduction of Theta functions and elliptic finite-gap potentials,
         Acta Appl. Math., 36  (1994) 87-117.

\bibitem{dS-AM-1883ab} C$^{\mathrm{te}}$ de Sparre,
            (On equation $\cdots$) Sur l'equation $\cdots$,
            Acta Math., 3 (1883) 105-140, 289-321.

\bibitem{D-CRASP-1882} G. Darboux,
         On a linear equation (Sur une \'equation lin\'eare),
         Comptes Rendus, 94  (1882) 1645-1648.

\bibitem{DN-JETP-1974} B.A. Dubrovin, S.P. Novikov,
         Periodic and conditionally periodic analogs of the many soliton solutions of the Korteweg-de Vries equation,
         Sov. Phys. JETP, 40 (1975) 1058-1063.

\bibitem{EK-AAM-1994} V.Z. Enol'skii, N.A. Kostov,
         On the geometry of elliptic solitons,
         Acta Appl. Math., 36 (1994) 57-86. 

\bibitem{81FA} A.S. Fokas, M.J. Ablowitz,
        Linearization of the Korteweg--de Vries and Painlev\'{e} II equations,
        Phys. Rev. Lett., 47 (1981) 1096-1099.

\bibitem{82FA} A.S. Fokas,  M.J. Ablowitz,
        Direct linearizations of the Korteweg-de Vries equations,
        in {Mathematical Methods in Hydrodynamics and Integrability in Dynamical Systems},
        Eds: M. Tabor, Y.M. Treve,
        AIP Conference Proceedings, 88 (1982) 237-241.

\bibitem{83-FA-KP} A.S. Fokas, M.J. Ablowitz,
        On the inverse scattering and direct linearizing transforms for the Kadomtsev-Petviashvili equation,
        Phys. Lett. A, 94 (1983) 67-70.

\bibitem{82Frobenius}  F.G. Frobenius,
         Ueber die elliptischen funktionen zweiter art,
         J. Reine Angew. Math., 93 (1882) 53-68.

\bibitem{18F}  W. Fu, Direct linearisation of the discrete-time two-dimensional Toda lattices,
        J. Phys. A: Math. Theor., 51 (2018) No.334001 (21pp).

\bibitem{17FN}  W. Fu, F.W. Njihoff,
        Direct linearizing transform for three-dimensional discrete integrable systems:
        The lattice AKP, BKP and CKP equations,
        Proc. R. Soc. A, 473 (2017) No.20160915 (22pp).

\bibitem{18-FuweiNijhoff} W. Fu,  F.W. Nijhoff,
        Linear integral equations, infinite matrices, and soliton hierarchies,
        J. Math. Phys., 59 (2018) No.071101 (28pp).

\bibitem{21FN}  W. Fu, F.W. Njihoff,
        On non-autonomous differential-difference AKP, BKP and CKP equations,
        Proc. R. Soc. A, 477 (2021) No.20200717 (20pp).

\bibitem{GW-MZ-1995} F. Gesztesy, R. Weikard,
        Treibich-Verdier potentials and the stationary (m)KdV hierarchy,
        Math. Z., 219 (1995) 451-476. 

\bibitem{GW-AM-1996}  F. Gesztesy, R. Weikard,
          Picard potentials and Hill's equation on a torus,
          Acta Math., 176 (1996) 73-107. 

\bibitem{GW-AM-1998}  F. Gesztesy, R. Weikard,
          A characterization of all elliptic algebro-geometric solutions of the AKNS hierarchy,
          Acta Math., 181 (1998) 63-108. 


\bibitem{11-Hie-BSQ} J. Hietarinta,
         Boussinesq-like multi-component lattice equations and multi-dimensional consistency,
         J. Phys. A: Math. Theor., 44   (2011) No.165204 (22pp).

\bibitem{11-HZ-SIGMA}	J. Hietarinta, D.J. Zhang,
        Soliton taxonomy for a modification of the lattice Boussinesq equation,
        SIGMA, 7  (2011) No.061 (14pp).

\bibitem{86-Hirota-PhyD} R. Hirota,
       Reduction of soliton equations in bilinear form,
       Physica D,  18 (1986) 161-170.

\bibitem{I-PRSE-1940}  E.L. Ince,
         Further investigations into the periodic Lam\'e functions,
         Proc. Roy. Soc. Edinburgh, 60 (1940) 83-99. 

\bibitem{23Kakei} S. Kakei,
         Solutions to the KP hierarchy with an elliptic background,  	arXiv: 2310.11679.

\bibitem{K-FAA-1980}  I.M. Krichever,
       Elliptic solutions of the K-P equation and integrable systems of particles,
       Funct. Anal. Appl., 14 (1980) 282-290. 

\bibitem{KBBT-AMST-book-1995}  I.M. Krichever, O. Babelon, E. Billey, M. Talon,
       Spin generalization of the Calogero-Moser system and the matrix KP equation,
       in Topics in Topology and Mathematical Physics,
       Ed. S.P. Novikov,
       Amer. Math. Soc. Transl. Ser. 2, Vol.170,
       Amer. Math. Soc., Providence, 1995, 83-119. 

\bibitem{74KM} E.A. Kuznetsov, A.V. Mikhailov,
          Stability of stationary waves in nonlinear weakly dispersive media,
          Sov. Phys. JETP, 40 (1974) 855-859.

\bibitem{23-LQZ} S.S. Li, C.Z. Qu, D.J. Zhang,
          Solutions to the SU$(N)$ self-dual Yang-Mills equation,
          Physica D, 453 (2023) No.133828 (17pp)

\bibitem{22LxZhang} X. Li, D.J. Zhang,
          Elliptic soliton solutions: $\tau$ functions, vertex operators and bilinear identities,
          J. Nonlinear Sci., 32 (2022) No.70 (53pp).

\bibitem{24LxZhang} X. Li, D.J. Zhang,	
          The Lam\'e functions and elliptic soliton solutions: Bilinear approach,
          in ``Recent Progress in Special Functions'', Contemporary Mathematics of AMS,
          807 (2024) 171-195.

 \bibitem{77MZBIM} S.V. Manakov, V.E. Zakharov, L.A. Bordag, A.R. Its, V.B. Matveev ,
          Two-dimensional solitons of the Kadomtsev-Petviashvili  equation and their interaction,
          Phys. Lett. A, 63 (1977) 205-206.

 \bibitem{MS-LMP-2006}  V.B. Matveev, A.O. Smirnov,
       On the link between the Sparre equation and Darboux-Treibich-Verdier equation,
       Lett. Math. Phys., 76 (2006) 283-295. 

\bibitem{Naka-LMP-2024} A. Nakayashiki,
          Vertex operators of the KP hierarchy and singular algebraic curves,
          Lett. Math. Phys., 114 (2024) No.50 (36pp).

\bibitem{85N} F.W. Nijhoff,
        Theory of integrable 3-dimensional nonlinear lattice equations,
        Lett. Math. Phys.,  9 (1985)  235-241.


\bibitem{IMRN2010}F.W. Nijhoff, J. Atkinson,
          Elliptic $N$-soliton solutions of ABS lattice equations,
          Int. Math. Res. Not., 2010 (2010) 3837-3895.

\bibitem{09NAH} F.W. Nijhoff,  J. Atkinson, J. Hietarinta,
          Soliton solutions for ABS lattice equations: I. Cauchy matrix approach,
          J. Phys. A: Math. Theor., 42 (2009) No.404005 (34pp).

\bibitem{85-NCW} F.W.  Nijhoff, H.W. Capel, G.L. Wiersma,
        Integrable lattice systems in two and three dimensions,
        in: Geometric Aspects of the Einstein Equations and Integrable Systems,
        Ed. R. Martini, Lect. Not. Phys., 239 (1985) 263-302.

\bibitem{92NPCQ} F.W. Nijhoff, V.G. Papageorgiou, H.W. Capel, G.R.W. Quispel,
        The lattice Gel'fand--Dikii hierarchy,
        Inverse Probl., 8 (1992)  597-621.

\bibitem{83-NQC} F.W. Nijhoff, G.R.W. Quispel, H.W. Capel,
        Direct linearisation of difference-difference equations,
        Phys. Lett. A, 97 (1983) 125-128.

\bibitem{83NQVC} F.W. Nijhoff, G.R.W. Quispel, J. van der Linden, H.W. Capel,
        On some linear integral-equations generating solutions of non-linear partial-differential equations,
        Physica A, 119 (1983)  101-142.

\bibitem{NSZ-2019} F.W. Nijhoff, Y.Y. Sun, D.J. Zhang,
         Elliptic solutions of Boussinesq type lattice equations and the elliptic $N$th  root of unity,
         Commun. Math. Phys., 399  (2023) 599-650.

\bibitem{82FVQC} F.W. Nijhoff, J. van der Linden, G.R.W. Quispel, H.W. Capel,
       Linearization of the nonlinear Schr\"odinger equation and the isotropic Heisenberg spin chain,
       Phys. Lett. A, 89 (1982) 106-108.

\bibitem{book-Pastras} G. Pastras,
       The Weierstrass Elliptic Function and Applications in Classical and Quantum Mechanics --
       A Primer for Advanced Undergraduates,
       Springer Nature Switzerland AG, 2020.

\bibitem{82QNC} G.R.W. Quispel, F.W. Nijhoff,  H.W. Capel,
       Linearization of the Boussinesq equation and the modified Boussinesq equation,
       Phys. Lett. A, 91 (1982) 143-145.

\bibitem{84QNCV} G.R.W. Quispel, F.W. Nijhoff, H.W. Capel, J. van der Linden,
        Linear integral-equations and nonlinear difference difference-equations,
        Physica A, 125 (1984) 344-380.

\bibitem{84SAF}  P.M. Santini, M.J. Ablowitz,  A.S Fokas,
        The direct linearization of a class of nonlinear evolution equations,
        J. Math. Phys., 25 (1984)  2614-2619.

\bibitem{S-MN-1989}  A.O. Smirnov,
       Elliptic solutions of the Korteweg-de Vries equation,
       Math. Notes., 45  (1989)  467-481. 

\bibitem{S-AAM-1994}  A.O. Smirnov,
       Finite-gap elliptic solutions of the KdV equation,
       Acta Appl. Math., 36 (1994)  125-166. 

\bibitem{S-RASSM-1995}  A.O. Smirnov,
        Elliptic solutions of the nonlinear Schr\"odinger equation and the modified Korteweg-de Vries equation,
        Russ. Acad. Sci. Sb. Math.   82 (1995) 461-470. 

\bibitem{S-CRMP-2002}  A.O. Smirnov,
       Elliptic solitons and Heun's equation,
       in The Kowalevski Property,
       Ed. V.B. Kuznetsov,
       CRM Proc. Lect. Notes., Vol.32, 
       Amer. Math. Soc., Providence, 2002, pp287-305. 

\bibitem{S-LMP-2006} A.O. Smirnov,
       Finite-gap solutions of the Fuchsian equations,
       Lett. Math. Phys., 76 (2006) 297-316. 


\bibitem{book-AbraSte}  T.H. Southard,
        \S 18. Weierstrass elliptic and related functions,
        in Handbook of Mathematical Functions with Formulas, Graphs, and Mathematical Tables,
        M. Abramowitz and I.A. Stegun (eds.), 9th printing, Dover, New York, 1972.

\bibitem{17-SunZhangNijhoff} Y.Y. Sun, D.J. Zhang, F.W.  Nijhoff,
        The Sylvester equation and the elliptic Korteweg-de Vries system,
        J. Math. Phys., 58 (2017) No.033504 (25pp).

\bibitem{T-CMP-2003}  K. Takemura,
       The Heun equation and the Calogero-Moser-Sutherland system I: The Bethe ansatz method,
       Commun. Math. Phys., 235 (2003) 467-494. 

\bibitem{T-DMJ-1989}  A. Treibich,
       Tangential polynomials and elliptic solitons,
       Duke Math. J.,  59  (1989) 611-627. 

\bibitem{T-AAM-1994}  A. Treibich,
       New elliptic potentials,
       Acta Appl. Math., 36 (1994) 27-48. 

\bibitem{T-DMJ-1997}  A. Treibich,
      Matrix elliptic solitons,
      Duke Math. J., 90 (1997) 523-547. 

\bibitem{T-RMS-2001}  A. Treibich,
       Hyperelliptic tangential covers and finite-gap potentials,
       Russ. Math. Surv., 56 (2001) 89-136. 

\bibitem{TV-G-book-1990}  A. Treibich, J.-L. Verdier,
       (Elliptic solitons) Solitons elliptiques,
       in The Grothendieck Festschrift, III.
       A Collection of Articles Written in Honor of the 60th Birthday of Alexander Grothendieck,
       Eds. P. Cartier, L. Illusie, N.M. Katz, et al.,
       Birkh\"auser-Boston, Boston, 1990, pp437-480.  

\bibitem{TV-CRASP-1990}  A. Treibich, J.-L. Verdier,
       (Tangential covers and sums of 4 triangular numbers)
       Rev$\hat{\mathrm{e}}$tements tangentiels et sommes de 4 nombres triangulaires,
       C. R. Acad. Sci. Paris, t. 311, Series I (1990) 51-54. 

\bibitem{V-AA-book-1988}  J.-L. Verdier,
       New elliptic solitons,
       in Algebraic Analysis, Vol.2, special volume dedicated to Prof. M. Sato on his 60th birthday,
       Eds. M. Kashiwara, T. Kawai,
       Academic Press, New York, 1988, pp901-910. 
       
\bibitem{V-LMP-2011}  A.P. Veselov, 
       On Darboux-Treibich-Verdier potentials, 
       Lett. Math. Phys., 96 (2011) 209-216. 

\bibitem{74Wahlqulst}   H.D. Wahlquist,
         B\"{a}cklund transformations of potentials of the Korteweg-de Vries equation
         and the interaction of solitons with conidal waves,
         In: B\"acklund Transformations, the Inverse Scattering Method, Solitons, and Their Applications,
         Ed. R.M. Miura,
         Springer-Verlag, Berlin, 1976, pp162-183.

\bibitem{14-XuZhangZhao} D.D. Xu, D.J. Zhang, S.L. Zhao,
        The Sylvester equation and integrable equations: I. The Korteweg-de Vries system and sine-Gordon equation,
          J. Nonlinear Math. Phys., 21 (2014)  382-406.

\bibitem{21YF} Y. Yin, W. Fu,
        Linear integral equations and two-dimensional Toda systems,
        Stud. Appl. Math., 147 (2021) 1146-1193.

\bibitem{2013-YN-JMP} S. Yoo-Kong, F.W. Nijhoff,
         Elliptic $(N, N^\p)$-soliton solutions of the lattice Kadomtsev-Petviashvili equation,
         J. Math. Phys., 54 (2013) No.043511 (20pp).

\bibitem{74ZS} V.E. Zakharov,  A.B. Shabat,
        A scheme for integrating the nonlinear equations of mathematical physics by the method of
        the inverse scattering problem. I.,
        Funct. Anal. Appl., 8  (1974) 226-235.

\bibitem{13ZZ} D.J. Zhang, S.L. Zhao,
        Solutions to the ABS lattice equations via generalized Cauchy matrix approach,
        Stud. Appl. Math., 131 (2013) 72-103.

\bibitem{12ZZN} D.J. Zhang, S.L. Zhao,  F.W. Nijhoff,
       Direct linearisation of the extended lattice BSQ system,
       Stud. Appl. Math., 129 (2012)  220-248.


%
%
%
%
%
%
%
%
%
%
%
%
%
%
%
%
%
%
%
%
%
%
%
%
%
%
%
%
%
%
%
%
%
%
%
%
%
%
%
%
%
%
%
%
%
%
%
%
%
%


\end{thebibliography}
\end{document}